\documentclass[aps,prx,reprint,twocolumn,notitlepage,superscriptaddress,longbiblirography]{revtex4-1}

\usepackage{amsmath,amssymb}
\usepackage{ulem}
\usepackage{color}
\usepackage{bm}
\usepackage{bbm}
\usepackage{mathtools}
\usepackage{hyperref}
\usepackage{amsthm}
\usepackage{ulem}
\usepackage{CJK}
\theoremstyle{plain}
\newtheorem{thm}{Theorem}
\newtheorem*{thm*}{Theorem}
\newtheorem{lmm}{Lemma}
\newtheorem{cor}{Corollary}
\newtheorem{definition}{Definition}

\usepackage[utf8]{inputenc}

\usepackage{graphicx}
\usepackage[svgnames]{xcolor}
\usepackage{float, subcaption, tikz}
\usepackage{tcolorbox}

\newcommand{\Dimerpara}{\vcenter{\hbox{\includegraphics[scale=0.07]{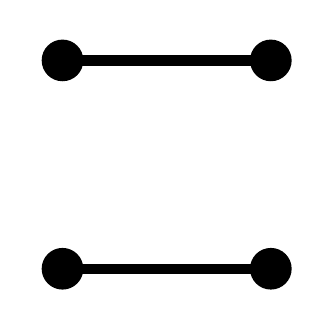}}}}
\newcommand{\Dimerperp}{\vcenter{\hbox{\includegraphics[scale=0.07]{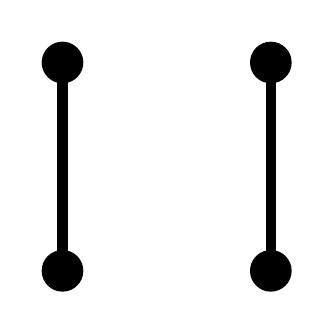}}}}

\let\Span\relax
\DeclareMathOperator{\Span}{Span}
\let\arccosh\relax
\DeclareMathOperator{\arccosh}{arccosh}
\let\gap\relax
\DeclareMathOperator{\gap}{gap}

\begin{document}
\title{
   Rigorous lower bound on dynamical exponents in gapless frustration-free systems
}
\author{Rintaro Masaoka}
\affiliation{Department of Applied Physics, The University of Tokyo, Tokyo 113-8656, Japan}
\author{Tomohiro Soejima (\begin{CJK*}{UTF8}{bsmi}副島智大\end{CJK*})}
\affiliation{Department of Physics, Harvard University, Cambridge, Massachusetts 02138, USA}
\author{Haruki Watanabe}\email{hwatanabe@g.ecc.u-tokyo.ac.jp}
\affiliation{Department of Applied Physics, The University of Tokyo, Tokyo 113-8656, Japan}
\date{\today}

\begin{abstract}
   This work rigorously establishes a universal lower bound $z\ge2$ for the dynamical exponent in frustration-free quantum many-body systems whose ground states exhibit power-law decaying correlations. The derivation relies on the Gosset-Huang inequality, providing a unified framework applicable across various lattice structures and spatial dimensions, independent of specific boundary conditions.
Remarkably, our result can be applied to prove new bounds for dynamics of classical stochastic processes.
Specifically, we utilize a well-established mapping from the time evolution of local Markov processes with detailed balance to that of frustration-free quantum Hamiltonians, known as Rokhsar-Kivelson Hamiltonians.
This proves $z \ge 2$ for such Markov processes, which is an improvement over existing bounds.
Beyond these applications, the quantum analysis of the $z\ge2$ bound is further broadened to include systems exhibiting hidden correlations, which may not be evident from purely local operators.
\end{abstract}

\maketitle

\section{Introduction}\label{sec: introduction}
The study of quantum many-body systems often relies on theoretically tractable models to gain fundamental insights into complex phenomena.
Frustration-free (FF) systems constitute a notable class of such models.
These systems are defined by Hamiltonians where the ground state simultaneously minimizes every local term of the Hamiltonian.
This property often simplifies the ground-state structure and facilitates theoretical analysis.
Historically, numerous FF models have been pivotal in advancing our understanding of novel quantum phases.
Prominent examples include the Affleck-Kennedy-Lieb-Tasaki model, which provided key insights into 1+1D symmetry-protected topological phases~\cite{affleckRigorousResultsValencebond1987}, and the Kitaev toric code, a solvable model for 2+1D topological order~\cite{kitaevFaulttolerantQuantumComputation2003}.
The remarkable ability of these fine-tuned FF systems to describe broader quantum phases is often attributed to the universality of these phases; specifically, the constraints imposed by frustration-freeness are presumed not to alter the universal properties of gapped quantum systems.

In contrast, gapless systems tell a different story.
Accumulating evidence indicates that the FF condition can profoundly influence the universal characteristics of gapless systems.
One of the most prominent features is a nontrivial dynamical exponent.
This exponent, often denoted $z$, is formally defined through the asymptotic behavior of the finite-size gap $\epsilon \sim L^{-z}$, where $L$ is the linear size of the system.
Many ordinary gapless systems, such as the critical point of the transverse field Ising model, or the Fermi liquid with finite chemical potential, exhibit $z=1$.
In contrast, all known gapless FF systems exhibit $z \ge 2$.
Paradigmatic Hamiltonians in this class include the ferromagnetic Heisenberg model and the quantum dimer model at the Rokhsar-Kivelson (RK) point, both of which are characterized by $z=2$.
However, a general proof of a bound $z \ge 2$ has been lacking.

Motivated by this observation, we prove $z \ge 2$ for FF Hamiltonians whose ground-state correlations decay no faster than a power law. Our result relies on a pivotal inequality in FF systems established by Gosset and Huang~\cite{gossetCorrelationLengthGap2016}, and holds irrespective of the choice of boundary conditions.
Our theorem applies to a wide class of gapless, FF systems (bosonic or fermionic) on arbitrary lattices in any spatial dimension.
Using our framework, we explicitly show that $z \ge 2$ holds for various representative gapless FF systems, including the generalized Rokhsar-Kivelson (RK) Hamiltonians at criticality, FF models exhibiting hidden correlations, and models with conserved quantities.
These findings suggest that the atypical dynamical exponents in gapless FF systems can be understood from a unified perspective based on ground-state correlation functions.

Our result puts an important constraint on the low-energy physics of FF systems. A large number of gapless systems arise as a quantum critical point between different phases. Many of these transitions are endowed with emergent relativistic conformal field symmetry. The symmetry includes Lorentz invariance, which forces $z=1$. Therefore, our bound shows FF systems cannot represent conformally invariant critical points, instead requiring us to understand non-relativistic critical points.

A notable consequence of this can be seen in translation-invariant systems. The dominant low-energy excitations in such systems often follow a dispersion relation of the form $E_k \sim |k|^z$.
The resulting quadratic or softer dispersion for $z \ge 2$ is crucial for enabling spontaneous $\mathrm{U}(1)$ symmetry breaking at zero temperature in 1+1 dimensions~\cite{watanabeCriticalSpontaneousBreaking2024}, which is a phenomenon forbidden in systems with $z=1$ under standard conditions for continuous symmetries~\cite{merminAbsenceFerromagnetismAntiferromagnetism1966,hohenbergExistenceLongRangeOrder1967}.

Remarkably, our result can also be used to prove a bound on \textit{classical} systems. We consider classical stochastic processes described by local Markov processes that satisfy the detailed balance condition. This includes a large class of important dynamics, such as the Glauber dynamics for the Ising model. A well-established mapping connects the transition-rate matrix of such processes to \textit{quantum} FF Hamiltonians called Rokhsar-Kivelson Hamiltonians. Our bound is, therefore, directly applicable to such Markov processes, providing a new bound on their dynamical exponents.

What is the consequence of a bound on $z$? The efficiency of Markov simulations can be measured by their autocorrelation time $\tau$, which quantifies the time it takes to be uncorrelated from a previous configuration. At a critical point, this time diverges as $\tau \propto L^z$, a phenomenon called critical slowdown.
Our new universal bound establishes a fundamental limit on the performance of such calculations.
More precisely, our results provide a no-go theorem: Under the assumptions of detailed balance and local state updates, one can never achieve $z<2$.
Standard methods such as Gibbs sampling (heat-bath) algorithms~\cite{gemanStochasticRelaxationGibbs1984} and the Metropolis-Hastings algorithms~\cite{metropolisEquationStateCalculations1953,hastingsMonteCarloSampling1970} belong to this class, and numerical studies indeed find $z\ge 2$ at criticality (see Table~\ref{table of dynamical exponent}).
There has been a long-standing effort to develop faster Markov chain Monte Carlo (MCMC) algorithms, focusing on nonlocal state updates~\cite{swendsenNonuniversalCriticalDynamics1987,wolffCollectiveMonteCarlo1989} or breaking the detailed balance condition~\cite{suwaMarkovChainMonte2010,turitsynIrreversibleMonteCarlo2011}.
Our no-go theorem provides a rigorous foundation for these works.

\begin{table}[t]
\centering
\begin{tabular}{ll}
   \hline \hline 
 Equilibrium states & Dynamical exponent $z$ \\ \hline
 Classical Ising (2D) & 2.1667(5)~\cite{nightingaleMonteCarloComputation2000}\\
 Classical Ising (3D) & 2.0245(15)~\cite{hasenbuschDynamicCriticalExponent2020} \\
 Classical Heisenberg (3D) & 2.033(5)~\cite{astilleroComputationDynamicCritical2019} \\
 Three-state Potts (2D) & 2.193(5)~\cite{muraseDynamicCriticalExponents2008} \\
 Four-state Potts (2D) & 2.296(5)~\cite{arashiroComparativeStudyDynamic2009}\\
 \hline \hline 
\end{tabular}
\caption{
 Numerical values of dynamical exponents for various critical points using the Gibbs sampling algorithms or the Metropolis-Hastings algorithm.
 For all examples in this list, $z \ge 2$ holds.
}
\label{table of dynamical exponent}
\end{table}

Let us now comment on the connection to existing literature.
The bound $z \ge 2$ was conjectured by Isakov et al. in Ref.~\cite{isakovDynamicsConformalQuantum2011} in the context of conformal quantum critical points, which form a subclass of gapless FF systems.
There are numerical confirmations of the lower bound $z \ge 2$ in many examples, such as the RK point of the quantum dimer model~\cite{rokhsarSuperconductivityQuantumHardCore1988,henleyRelaxationTimeDimer1997,ardonneTopologicalOrderConformal2004}, ferromagnetic Heisenberg models~\cite{heisenbergZurTheorieFerromagnetismus1928,blochZurTheorieFerromagnetismus1930}, and many other models~\cite{isakovDynamicsConformalQuantum2011, chenGaplessQuantumSpin2017,kumarMulticriticalTopologicalTransition2021,watanabeCriticalSpontaneousBreaking2024,saitoExactMatrixProduct2024}.
The bound $z \ge 2$ can be shown for FF systems with Nambu-Goldstone modes using a variational state, and this result has been generalized to systems with quasi-Nambu-Goldstone modes~\cite{ogunnaikeUnifyingEmergentHydrodynamics2023,renQuasiNambuGoldstoneModesManybody2024}.
It has also been observed numerically that the dynamical exponent changes from $z\ge2$ to $z=1$ when a perturbation violates the FF property~\cite{isakovDynamicsConformalQuantum2011,kumarMulticriticalTopologicalTransition2021,watanabeCriticalSpontaneousBreaking2024}.

Partial proofs of $z \ge 2$ focusing on the models under open boundary conditions without boundary Hamiltonian terms have been obtained by establishing upper bounds on the spectral gap of FF Hamiltonians~\cite{gossetLocalGapThreshold2016, anshuImprovedLocalSpectral2020,lemmQuantitativelyImprovedFinitesize2022,lemmCriticalFinitesizeGap2024,knabeEnergyGapsElementary1988}.
Let us briefly outline the core argument.
First, consider a FF Hamiltonian $\hat{H}^{(L)} = \sum_{i \in \Lambda_L} \hat{H}_i$, defined on a lattice $\Lambda_L$ of linear size $L$.
We assume its spectral gap, $\text{gap}(\hat{H}^{(L)})$, vanishes in the thermodynamic limit $L \to \infty$.
Next, consider the Hamiltonian $\hat{H}^{(l)} = \sum_{i \in \Lambda_l \subset \Lambda_L} \hat{H}_i$, which is restricted to a subregion $\Lambda_l$ with linear size $l$.
In the limit $L \to \infty$, it has been shown that the gap of $\hat{H}^{(l)}$ is bounded as
\begin{equation}
    \text{gap}(\hat{H}^{(l)}) \le \frac{A}{l^2},
\end{equation}
where $A$ is a positive constant, which directly implies that  $z \ge 2$.
This type of inequality was first derived by Gosset and Mozgunov~\cite{gossetLocalGapThreshold2016}, improving the local gap threshold by Knabe~\cite{knabeEnergyGapsElementary1988}.
Since then, several variations and refinements of this inequality have been developed in subsequent studies~\cite{anshuImprovedLocalSpectral2020,lemmQuantitativelyImprovedFinitesize2022,lemmCriticalFinitesizeGap2024}.

While these results offer valuable insights and are widely applicable to FF models, their reliance on open boundary conditions (without boundary terms) presents a crucial limitation.
This specific setup is a consequence of requiring $\hat{H}^{(l)}$ to be a part of a larger Hamiltonian $\hat{H}^{(L)}$.
Consequently, such arguments cannot distinguish between spectral gaps arising from bulk excitations and those originating from edge modes~\cite{watanabeCriticalSpontaneousBreaking2024}.
This distinction is crucial because the dynamical exponent $z$ can, in principle, depend on the boundary conditions.
Indeed, there are known examples of quantum systems where $z$ varies with the choice of boundary conditions~\cite{chenGaplessQuantumSpin2017,verresenGaplessTopologicalPhases2021}.
Our proof circumvents this challenge, since it is applicable to periodic boundary conditions.

The rest of the paper is organized as follows.
In Sec.~\ref{sec: general framework}, we introduce a general framework for proving $z \ge 2$ in gapless FF systems, based on the Gosset-Huang inequality.
In Sec.~\ref{sec: RK}, we apply this framework to a specific class of FF models constructed from classical statistical mechanics, known as RK Hamiltonians.
We provide a rigorous proof that $z \ge 2$ holds for RK Hamiltonians of classical critical points, which are also known as conformal quantum critical points.
Furthermore, leveraging the correspondence between RK Hamiltonians and Markov processes, we elucidate the implications of the established bound $z \ge 2$ in classical nonequilibrium statistical mechanics and computational physics.
In Sec.~\ref{sec: hidden}, we discuss an extension of the Gosset-Huang inequality to handle not only excitations generated by local operators but also arbitrary excitations with localized energy, even if they are not generated by local operators.
We discuss examples of hidden correlations arising from nonlocal excitations but with localized energy density.
We show that even FF systems with product ground states can exhibit a power-law decay of correlation functions with system size and satisfy $z \ge 2$.
Finally, in Sec.~\ref{sec: discussion}, we summarize our findings and the implications of the established bound.
In the appendix, we give detailed explanations that we omit in the main text.

\section{$z \ge 2$ for FF systems with power-law decaying correlation functions}
\label{sec: general framework}
In this section, we prove our main theorem, which establishes the lower bound $z\ge2$ on the dynamical exponent for critical gapless FF systems.
The proof relies on an inequality by Gosset and Huang~\cite{gossetCorrelationLengthGap2016}, which we formulate carefully below.

\subsection{Definitions}

In this work, we study quantum systems on a $d$-dimensional lattice $\Lambda$ which is embedded in the $d$-dimensional Euclidean space or torus.
The Hilbert space on $\Lambda$ is given by
\begin{align}
 \mathcal{H} = \bigotimes_{x \in \Lambda} \mathcal{H}_{x},
\end{align}
where $\mathcal{H}_x$ is a finite-dimensional local Hilbert space.
Although the tensor product structure may no longer hold when considering subspaces that satisfy constraint conditions, subsequent discussions remain valid even in such cases.

The Hamiltonian is a Hermitian operator acting on this Hilbert space.
Without loss of generality, we assume that the Hamiltonians are positive semidefinite, and their lowest eigenvalue is normalized to zero.

\begin{definition}[Frustration-free Hamiltonians]
 A Hamiltonian $\hat H$ is FF if it can be decomposed as
    \begin{align}
        \hat H = \sum_{i \in \mathcal{S}} \hat H_i,
    \end{align}
 where each term $\hat H_i$ is positive semidefinite and there exists a ground state $|\Psi\rangle$ satisfying  $\hat H_i |\Psi\rangle = 0$ for all $i \in \mathcal{S}$ (which implies that all ground states satisfy this condition).
\end{definition}

However, to obtain a meaningful class of Hamiltonians, the decomposition must be restricted because any Hamiltonian can be made trivially frustration-free via the decomposition $\hat H = \hat H$.
In this paper, we assume the strict locality of Hamiltonians.
In other words, we assume that each local term $\hat H_i$ is associated with a spatial position $\bm{x}_i$, and two terms $\hat H_i$ and $\hat H_j$ commute with each other if two positions $\bm{x}_i$ and $\bm{x}_j$ are sufficiently distant.

We now introduce the dynamical exponent $z$.
We use the definition of $z$ based on the spectral gap of the Hamiltonian and the system size.

\begin{definition}[Spectral gap]
 For a positive semidefinite operator  
   $\hat H$ with its lowest eigenvalue normalized to zero, the spectral gap, $\gap(\hat H)$, is its smallest nonzero eigenvalue.
\end{definition}

If the spectral gap of an FF Hamiltonian converges to zero in the large volume limit, we say that this Hamiltonian is gapless.

Note that this definition considers only the first excited state, differing from the more common definition of gapless systems, which requires a dense spectrum of excitations near the ground-state energy.
To our knowledge, no examples of local FF systems exist where these two definitions yield different outcomes~\cite{masaokaQuadraticDispersionRelations2024}.

For gapless FF systems, we define the dynamical exponent as follows.

\begin{definition}[dynamical exponent]
 Let $L$ be the linear size of the lattice measured in its embedding in Euclidean space or a torus.
 The dynamical exponent $z$ is defined by the relation
    \begin{align}
        \gap (\hat H_L) \sim L^{-z},
    \end{align}
\label{def:dynamic_critical_exponent}
where $\hat H_L$ is the Hamiltonian defined on a lattice of linear size $L$.
\end{definition}
We later define the system size more precisely in terms of the interaction graph.
While logarithmic corrections to this scaling relation can occur in general, we neglect them in our definition of $z$.

In translationally invariant gapless systems, the dynamical exponent is related to the dispersion relation $\omega(k) \sim k^z$ where $\omega(k)$ is the energy of an excitation with momentum $k$, provided these excitations are the dominant low-energy eigenstates.
From a renormalization-group perspective, a fixed point with dynamical exponent $z$ is invariant under the Lifshitz scale transformation $\bm{x} \to \lambda\bm{x}$, $t \to \lambda^z t$.

\subsection{Gosset-Huang inequality}\label{sec: GH inequality}

The lower bound $z \ge 2$ can be established without open boundary conditions by using the Gosset-Huang inequality~\cite{gossetCorrelationLengthGap2016}.
 We present a refined version of this inequality with an improved exponential coefficient, with the proof provided in Appendix~\ref{sec: proof of GH inequality}. A further generalization that removes unnecessary locality assumptions is discussed in Sec.~\ref{sec: hidden}.

\subsubsection{Interaction graph}
For the formulation and the proof of the Gosset-Huang inequality, it is useful to define the interaction graph.
The interaction graph for a Hamiltonian $\hat H = \sum_{i \in \mathcal{S}} \hat H_i$ has $|\mathcal{S}|$ vertices labeled by $i \in \mathcal{S}$ and there is an edge between vertex $i$ and vertex $j$ when $[\hat H_i,\hat H_j] \ne 0$.

When all local terms commute, i.e., for a stabilizer Hamiltonian with
$\bigl[\hat H_i,\hat H_j\bigr]=0$ for every $i,j\in\mathcal{S}$, the interaction graph defined above has no edges. 
Since these systems are always gapped, we ignore them in the following analysis.

An example of the interaction graph is shown in Fig.~\ref{fig: original and interaction graph}.
\begin{figure}[t]
    \centering
    \begin{minipage}{0.4\hsize}
        \centering
        \includegraphics[width=0.7\hsize]{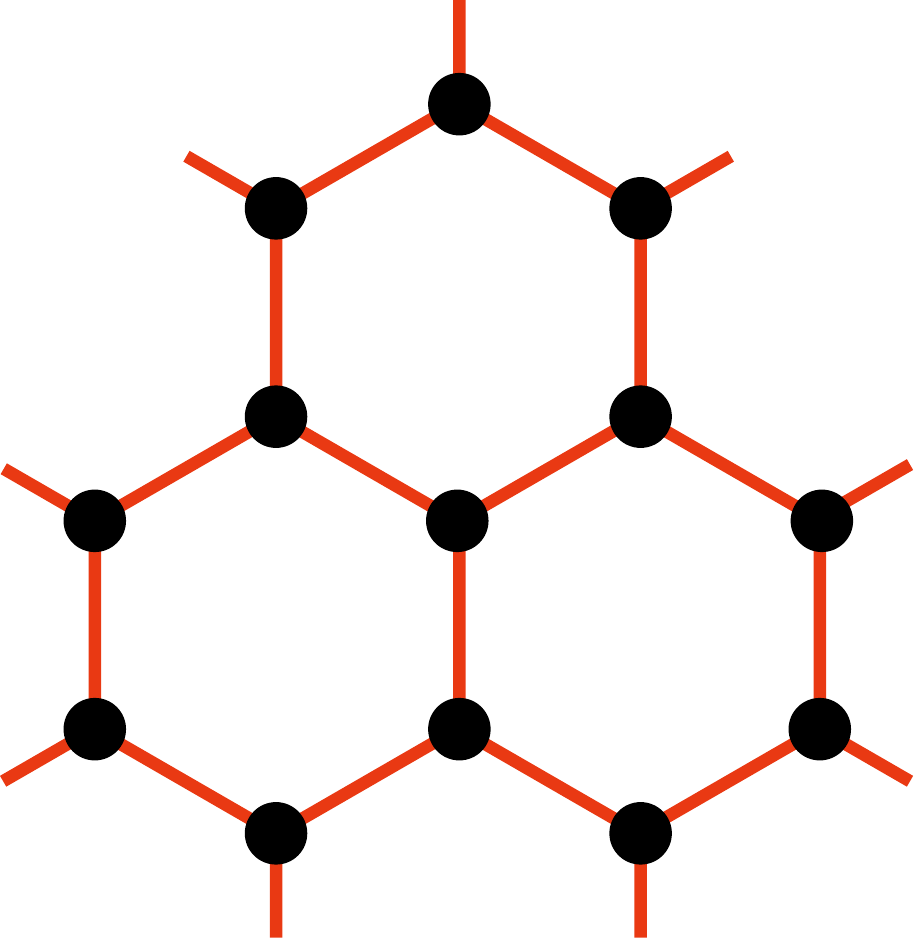}
        \subcaption{\label{fig: original graph}}
    \end{minipage}
    \begin{minipage}{0.4\hsize}
        \centering
        \includegraphics[width=0.7\hsize]{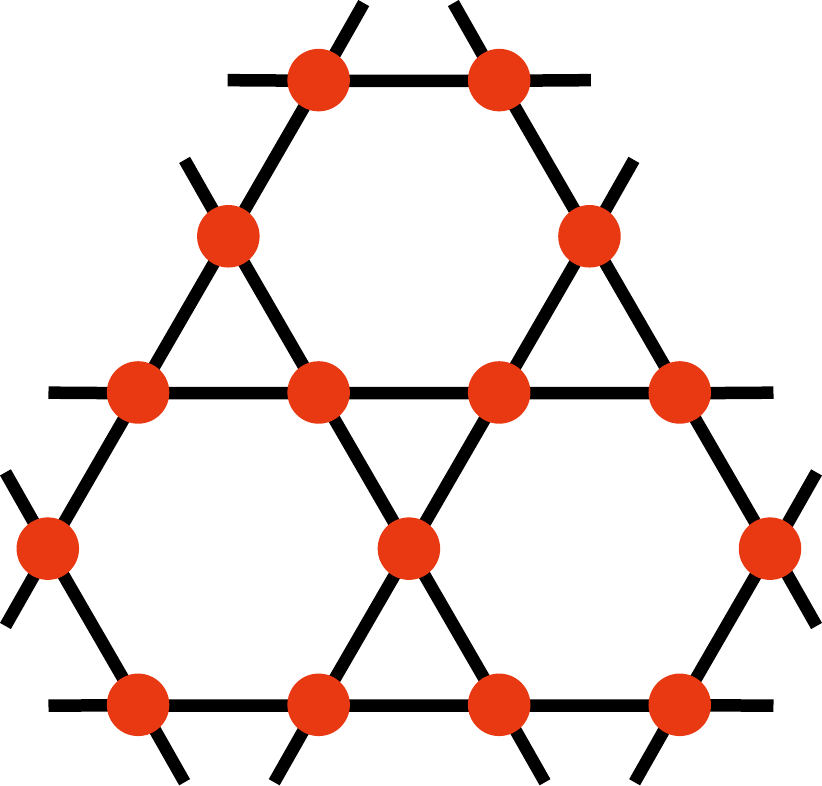}
        \subcaption{\label{fig: interaction graph}}
    \end{minipage}
    \caption{
      An example of the interaction graph. (a) The physical lattice (a honeycomb lattice) on which a Hamiltonian is defined.
      Black vertices represent spins, and red edges represent nearest-neighbor interactions.
      (b) The interaction graph of the Hamiltonian, where red vertices correspond to the local terms of the Hamiltonian, and black edges represent nonzero commutators between two local terms. Here, we assume that interaction terms sharing a spin do not commute.
    }
   \label{fig: original and interaction graph} 
\end{figure}
If two vertices $i$ and $j$ share an edge, we say that $i$ is adjacent to $j$.
The degree $g_i$ of a vertex $i$ is given by the number of vertices adjacent to $i$:
\begin{align}
 g_i \coloneqq |\{j \in \mathcal{S} \mid [\hat H_j, \hat H_i] \ne 0\}|.
\end{align}
The maximum degree $g$ of the graph is defined by
\begin{align}
 g \coloneqq\max_{i \in \mathcal{S}} g_i.
\end{align}
Another integer that characterizes the interaction graph is the chromatic number $c$.
It is the smallest number of colors needed for coloring the vertices so that no two adjacent vertices have the same color.
For example, the bipartite graph has $c = 2$ by definition, and all planar graphs satisfy $c \le 4$ from the four-color theorem.
The chromatic number $c$ is finite if the maximum degree $g$ is finite, since the greedy coloring algorithm ensures $c \le g+1$.

We can naturally define a distance $d(i, j)$ between two local terms $\hat H_i$ and $\hat H_j$ of the Hamiltonian by the graph theoretical distance, which is the minimum length of a path connecting two vertices $i$ and $j$ on the graph. 
If no path connects two vertices $i$ and $j$, we formally set $d(i, j) = \infty$.

We further define the distance between two operators $\hat{\mathcal{O}}$ and $\hat{\mathcal{O}}'$ by
\begin{align}
 D(\hat{\mathcal{O}}, \hat{\mathcal{O}}') \coloneqq \min \{d(i, j) \mid [\hat{\mathcal{O}}, \hat H_i] \ne 0, [\hat{\mathcal{O}}', \hat H_j] \ne 0 \}.
\end{align}
Strictly speaking, $D(\hat{\mathcal{O}}, \hat{\mathcal{O}}')$ is not a distance in the mathematical sense, since it does not satisfy the triangle inequality, although $d(i, j)$ does.

 Note that we formally set $D(\hat{\mathcal{O}}, \hat{\mathcal{O}}') = \infty$ when $\hat{\mathcal{O}}$ commutes with all local terms or no path connects $\hat{\mathcal{O}}$ and $\hat{\mathcal{O}}'$ in the interaction graph.

For a local Hamiltonian with a connected interaction graph, we can define the linear size of the system by $L =\max_{i,j} d(i, j)$, which is proportional to the linear size measured in the embedded space.

\subsubsection{Formulation of the Gosset-Huang inequality}

Let $\hat H$ be an FF Hamiltonian with a decomposition $\hat H = \sum_i \hat H_i$.
We define the dimensionless spectral gap by
\begin{align}
 \epsilon \coloneqq \frac{\gap(\hat H)}{\max_i \|\hat H_i\|},
    \label{dimless gap}
\end{align}
where $\|\cdot\|$ is the operator norm.
Then, the following inequality holds~\cite{gossetCorrelationLengthGap2016}:
\begin{align}&
 \frac{|\langle\Psi|\hat{\mathcal{O}}(\hat{\mathbbm{1}}-\hat{G})\hat{\mathcal{O}}'|\Psi\rangle|}{\|\hat{\mathcal{O}}^\dagger|\Psi\rangle\| \|\hat{\mathcal{O}}'|\Psi\rangle\|} \nonumber\\
   &
 \le 2\exp\left(-\left(\frac{D(\hat{\mathcal{O}}, \hat{\mathcal{O}}')-1}{c-1}-2\right)\sqrt{\frac{\epsilon}{g^2+\epsilon}}\right),
   \label{original GH inequality}
\end{align}
where $g$ is the maximum degree of the interaction graph and $c$ is the chromatic number of the interaction graph.
We call this inequality the Gosset-Huang inequality.
This inequality shows that the correlation function on the left-hand side is bounded by an exponential in the graph distance $D(\hat{\mathcal{O}},\hat{\mathcal{O}}')$.
If the ground state is unique, the left-hand side of Eq.~\eqref{original GH inequality} is the connected correlation function, where the ground-state projector is given by $\hat{G} = | \Psi \rangle \langle \Psi |$.
We also note that the left-hand side vanishes if $D(\hat{\mathcal{O}},\hat{\mathcal{O}'}) = \infty$. 

It is illuminating to compare the Gosset-Huang inequality with the bound
\begin{align}
      \frac{|\langle\Psi|\hat{\mathcal{O}}(\hat{\mathbbm{1}} - \hat{G})\hat{\mathcal{O}}'|\Psi\rangle|}{\|\hat{\mathcal{O}}\| \|\hat{\mathcal{O}}'\|}
 \le Ce^{-D(\hat{\mathcal{O}}, \hat{\mathcal{O}}')\epsilon/2v},
      \label{Hastings inequality}
\end{align}
for positive constants $C$ and $v$,
which was derived in Ref.~\cite{hastingsLocalityQuantumMarkov2004} using the Lieb-Robinson bound~\cite{liebFiniteGroupVelocity1972}.
This inequality applies to generic (not necessarily FF) quantum systems.
The key difference between the two bounds is that Eq.~\eqref{original GH inequality} involves $\sqrt{\epsilon}$, whereas Eq.~\eqref{Hastings inequality} is linear in the gap $\epsilon$ in the exponent.
This square-root scaling is essential for establishing the lower bound $z \ge 2$ in gapless FF systems.

\subsubsection{A heuristic proof of $z\ge 2$}

The Gosset-Huang inequality~\cite{gossetCorrelationLengthGap2016} directly implies a relationship $\xi = O(1/\sqrt{\epsilon})$ between the correlation length $\xi$ and the spectral gap $\epsilon$ for FF Hamiltonians.
The correlation length $\xi$ is defined via the exponential decay of correlations:
\begin{align}
 \frac{|\langle\Psi|\hat{\mathcal{O}}(\hat{\mathbbm{1}}-\hat G)\hat{\mathcal{O}'}|\Psi\rangle|}{\|\hat{\mathcal{O}}^\dagger|\Psi\rangle\| \|\hat{\mathcal{O}}'|\Psi\rangle\|} \approx \mathrm{const} \times {{e}}^{-D(\hat{\mathcal{O}}, \hat{\mathcal{O}}')/\xi},
\end{align}
where $D(\hat{\mathcal{O}}, \hat{\mathcal{O}}')$ is the distance between the two operators.
This relation suggests $z \ge 2$, if the dynamical exponent $z$ is defined in the thermodynamic limit through the scaling relation
\begin{align}
 \epsilon(J) \sim \xi(J)^{-z},
\end{align}
where $\epsilon(J)$ is the spectral gap, $\xi(J)$ is the ground-state correlation length, and $J$ is a parameter approaching a gapless critical point $J_c$.
However, this argument relies on the existence of a one-parameter family of Hamiltonians, $\hat H(J)$, that remains frustration-free across the entire path toward the gapless point $J_c$.
To make the argument fully rigorous, one must demonstrate that such a one-parameter family exists and that the critical exponent is independent of the particular path taken toward the gapless point.
To avoid these subtle issues, we establish $z\ge 2$ using the definition of the dynamical exponent based on the finite-size scaling of the spectral gap (Definition \ref{def:dynamic_critical_exponent}), which is determined solely at the gapless point.

\subsection{A proof of $z \ge 2$ for FF systems with algebraic correlation functions}
\label{sec: proof of bound}

Let us now prove our main theorem. Our theorem establishes the lower bound $z\ge2$ under specific conditions related to long-range correlations in a ground state.

Let $\hat H = \sum_i \hat{H}_i$ be an FF Hamiltonian defined on any lattice and let $| \Psi \rangle$ denote a ground state.
We define the linear size $L$ of the system by $L = \max_{i,j}d(i,j)$.

\begin{thm} \label{thm: main}
Assume there exist local operators $\hat{\mathcal{O}}$ and $\hat{\mathcal{O}}'$ such that their graph distance satisfies $D(\hat{\mathcal{O}}, \hat{\mathcal{O}}') \ge AL$ for some constant $A>0$, and
 \begin{align}
 \frac{|\langle\Psi|\hat{\mathcal{O}}(\hat{\mathbbm{1}}-\hat{G})\hat{\mathcal{O}}'|\Psi\rangle|}{\|\hat{\mathcal{O}}^\dagger|\Psi\rangle\| \|\hat{\mathcal{O}}'|\Psi\rangle\|} \ge BL^{-\alpha},
 \label{assumption of main thm}
 \end{align}
 where  $\hat G$ denotes the projector onto the ground space of $\hat H$, and $\alpha \ge 0$ and $B > 0$ are constants independent of $L$.
 Then, for sufficiently large $L$, the spectral gap $\epsilon_L$ of the Hamiltonian satisfies
\begin{align}
      \epsilon_L \le  2\left(\frac{g(c-1)(\alpha+1)}{A}\right)^2\frac{(\ln L)^2}{L^2}.
\end{align}
Consequently, the dynamical exponent obeys the lower bound $z \ge 2$.
\end{thm}

 In our argument, the ground-state correlation function may contain two distinct contributions given by
   \begin{align}
      \frac{|\langle\Psi|\hat{\mathcal{O}}(\hat{\mathbbm{1}}-\hat{G})\hat{\mathcal{O}}'|\Psi\rangle|}{\|\hat{\mathcal{O}}^\dagger|\Psi\rangle\| \|\hat{\mathcal{O}}'|\Psi\rangle\|}
 = BD(\hat{\mathcal{O}}, \hat{\mathcal{O}}')^{-\alpha} + B'L^{-\alpha'}.
   \end{align}
 One contribution decays algebraically with the operator distance, and another term scales with the inverse system size yet is independent of the distance. Therefore, our results apply not only to critical FF systems that exhibit power-law correlations in the usual sense, but also to a broader class of gapless FF systems.

 Furthermore, as explained in Sec.~\ref{sec: hidden}, an even wider range of models can be addressed by employing the extended Gosset-Huang inequality and considering more general correlation functions that are not necessarily generated by local operators.
\begin{proof}
Combining the assumption in Eq.~\eqref{assumption of main thm} and the Gosset-Huang inequality, we have
\begin{align}
 BL^{-\alpha}
 & \le \frac{|\langle\Psi|\hat{\mathcal{O}}(\hat{\mathbbm{1}}-\hat{G})\hat{\mathcal{O}}'|\Psi\rangle|}{\|\hat{\mathcal{O}}^\dagger|\Psi\rangle\| \|\hat{\mathcal{O}}'|\Psi\rangle\|} \nonumber\\
    &
 \le 2 \exp\left[-\left(\frac{AL}{c-1}-2\right)\sqrt{\frac{\epsilon_L}{g^2+\epsilon_L}}\right]\nonumber\\
    &
 \le 2e^2 \exp\left(-\frac{AL}{c-1}\sqrt{\frac{\epsilon_L}{g^2+\epsilon_L}}\right),
\end{align}
where we used $\epsilon_L/(g^2+\epsilon_L) \le 1$.
Thus, we obtain
\begin{align}
   \frac{\epsilon_L}{g^2+\epsilon_L} &
   \le \left(\frac{c-1}{AL}\ln\frac{2e^2}{BL^{-\alpha}}\right)^2 \nonumber\\ &
 = \left(\frac{c-1}{A}\right)^2\left(\alpha\frac{\ln L}{L} + \frac{\ln(2e^2/B)}{L}\right)^2 \nonumber\\ &
    \eqqcolon f(L).
\end{align}
Let us choose a length $L_0$ sufficiently large so that $L_0 \ge 2e^2/B$, $L_0 \ge B/2e^2$ and $f(L_0) \le 1/2$.
For $L \ge L_0$, we have
\begin{align}
   \epsilon_L &\le \frac{g^2f(L)}{1-f(L)}
   \nonumber\\
   &
   \le \frac{g^2f(L)}{1-f(L_0)} \le 2g^2f(L) \nonumber\\
   &
   \le 2\left(\frac{g(c-1)}{A}\right)^2\left(\alpha\frac{\ln L}{L} + \frac{\ln L_0}{L}\right)^2 \nonumber\\
   &
   \le 2\left(\frac{g(c-1)(\alpha+1)}{A}\right)^2\frac{(\ln L)^2}{L^2}.
\end{align}
If $\alpha = 0$, this bound improves to
\begin{align}
 \epsilon_L \le 2g^2f(L) = 2\left(\frac{g(c-1)}{A}\ln\frac{2e^2}{B}\right)^2\frac{1}{L^2}.
\end{align}
These bounds on the spectral gap establish a lower bound of $z \ge 2$.
\end{proof}

We note that our theorem needs the existence of algebraic correlation functions, which is a sufficient condition for the gapless property, but not necessary.

We also note that Eq.~\eqref{Hastings inequality} gives a weaker bound $z \ge 1$ in general quantum systems defined by a local Hamiltonian.

\subsection{Example: $XX$ model with fine-tuned magnetic field}
One of the simplest examples of FF systems exhibiting power-law correlations is the $XX$ model with a fine-tuned magnetic field.
Place spin-$1/2$ variables, with computational basis $\{|0_i\rangle,|1_i\rangle\}$, on the sites $i\in \Lambda$ of an arbitrary $d$-dimensional lattice.  
The nearest-neighbor projector $\hat H_{\langle i,j \rangle}$ is defined by 
\begin{align}&
 \ker \hat H_{\langle i,j \rangle} = \Span\left\{|0_i0_{j}\rangle, |0_i1_{j}\rangle+|1_i0_{j}\rangle\right\} \label{kernel of XX+MF}.
\end{align}
This local term can be expressed in terms of Pauli operators as
\begin{align}
   \hat H_{\langle i,j \rangle} = \frac1{4}(-\hat\sigma_i^x\hat\sigma_j^x - \hat\sigma_i^y\hat\sigma_j^y-\hat\sigma^z_i-\hat\sigma^z_j)+\frac1{2}\hat{\mathbbm{1}}.
   \label{Hamiltonian of XX+MF}
\end{align}
The total Hamiltonian is given by $\hat H \coloneqq \sum_{e}\hat H_e$, where $e$ runs over all edges of the lattice.
If the underlying lattice is connected, the ground states of this Hamiltonian are given by
\begin{align}
 |\Psi_0\rangle \coloneqq |0\cdots0\rangle,\quad
 |\Psi_1\rangle \coloneqq \frac1{\sqrt{|\Lambda|}}\sum_{i \in \Lambda}\hat\sigma^-_i |\Psi_0\rangle.
 \label{GS of XXZ+DM+MF}
\end{align}
This Hamiltonian is FF, and the following correlation function decays as a power of the system size, not the distance between operators:
\begin{align}
 \frac{|\langle\Psi_0|\hat\sigma_i^+(\hat{\mathbbm{1}}-\hat G)\hat\sigma_j^-|\Psi_0\rangle|}{\|\hat\sigma^-_i|\Psi_0\rangle\| \|\hat\sigma^-_j|\Psi_0\rangle\|} 
 &
 = |\langle\Psi_0|\hat\sigma_i^+|\Psi_1\rangle\langle\Psi_1|\hat\sigma_j^-|\Psi_0\rangle|\nonumber\\
 &
 = \frac1{|\Lambda|}.
 \label{correlation function for XXZ + MF}
\end{align}
Here, we assumed $i \ne j$.
Therefore, we can use Theorem~\ref{thm: main} by setting $D(\hat{\sigma}^+_i,\hat{\sigma}^-_j) = \mathrm{const}\times L$ to obtain $z \ge 2$.

\section{Rokhsar-Kivelson Hamiltonians and Markov processes}
\label{sec: RK}
In this section, we show that our theorem is applicable to a large class of \textit{classical} statistical processes, thereby proving a new bound on their dynamical exponents.
To do so, we introduce a specific class of FF Hamiltonians, known as generalized RK Hamiltonians or simply RK Hamiltonians, which is constructed from Markov processes satisfying the detailed balance condition~\cite{henleyClassicalQuantumDynamics2004,castelnovoQuantumMechanicsClassical2005}.
This class of Hamiltonians is also known as stochastic matrix form Hamiltonians or stoquastic FF Hamiltonians.
RK Hamiltonians can be viewed as a discretized counterpart of the stochastic quantization in quantum field theories~\cite{parisiPERTURBATIONTHEORYGAUGE1981,dijkgraafRelatingFieldTheories2010}.

RK Hamiltonians provide many paradigmatic FF models such as the kinetic Ising model, Kitaev's toric code~\cite{kitaevFaulttolerantQuantumComputation2003}, and the original RK point in the quantum dimer model~\cite{rokhsarSuperconductivityQuantumHardCore1988}.
In particular, the RK Hamiltonians of classical critical points are known as conformal quantum critical points~\cite{ardonneTopologicalOrderConformal2004}.

Using the correspondence between (local) RK Hamiltonians and (local) Markov processes with the detailed balance condition, we prove the lower bound $z \ge 2$ for RK Hamiltonians of critical points and corresponding Markov processes.

\subsection{Review of construction of RK Hamiltonians from Markov processes}
In this subsection, we review the mapping between RK Hamiltonians and the transition rate for Markov processes. This mapping allows us to apply bounds on quantum FF Hamiltonians to classical statistical processes.

\subsubsection{Markov processes}

A classical statistical system is defined by a positive function $w(C)$, called the Boltzmann weight, assigned to each classical configuration $C \in \mathcal{C}$.
The expectation value of a physical quantity $\mathcal{O}$ in the canonical distribution is defined by
\begin{align}
\langle\mathcal{O}\rangle \coloneqq \sum_{C \in \mathcal{C}} \mathcal{O}(C) \frac{w(C)}{Z},
\label{expectation of Markov process}
\end{align}
where the partition function $Z$ is written as
\begin{align}
 Z \coloneqq \sum_{C \in \mathcal{C}} w(C).
\end{align}
The Markov process that relaxes toward the equilibrium state is governed by the master equation
\begin{align}
 \frac{\mathrm{d}}{\mathrm{d}t}p(C, t) = \sum_{C'} W_{CC'} p(C', t),
 \label{master eq}
\end{align}
where $p(C, t)$ denotes the probability of finding the system in configuration $C$ at time $t$.
We can also construct a Markov chain by discretizing the time $t$.
Then, the time evolution of the probability is given by
\begin{align}
 p(C, t+\delta t) = p(C,t) + \delta t \sum_{C'} W_{CC'}p(C', t),
    \label{discrete master eq}
\end{align}
where $\delta t$ is the time step.
The transition-rate matrix $W_{CC'}$ must satisfy
$W_{CC'}\geq 0$
when $C \ne C'$ and the conservation of probability
\begin{align}
\sum_{C \in \mathcal{C}} W_{CC'}=0.
\end{align}
We also assume the detailed balance condition
\begin{align}
 W_{CC'}w(C')  = W_{C'C}w(C),
\end{align}
so that $p_\mathrm{eq}(C) \coloneqq w(C)/Z$ is a steady state solution of the Master equation.
This property is written as
\begin{align}
 \sum_{C' \in \mathcal{C}} W_{CC'}w(C') = 0,
\end{align}
and called the balance condition.
This follows from conservation of probability and the detailed balance condition.
 
\subsubsection{Rokhsar-Kivelson Hamiltonians}
Let us construct RK Hamiltonians~\cite{rokhsarSuperconductivityQuantumHardCore1988,henleyClassicalQuantumDynamics2004,castelnovoQuantumMechanicsClassical2005} from local Markov processes with the detailed balance condition.
We consider the Hilbert space spanned by orthogonal states $|C\rangle$ labeled by classical configurations $C \in \mathcal{C}$ of a given Markov process.
The transition-rate matrix $W_{CC'}$ is described as
the operator $\hat W = \sum_{C,C \in \mathcal{C}}W_{CC'}|C\rangle\langle C'|$ acting on $\mathcal{H}$.
The quantum Hamiltonian corresponding to the transition rate $W_{CC'}$ is given by the following similarity transformation.
\begin{align}
   \hat H \coloneqq -\hat S^{-1}\hat W\hat S,
   \label{RK-MC map}
\end{align}
where $\hat S$ is a positive diagonal operator defined as
\begin{align}
 \langle C|\hat{S}|C'\rangle \coloneqq \sqrt{\frac{w(C)}{Z}} \delta_{CC'}.
   \label{def of S}
 \end{align}
This Hamiltonian is called a generalized RK Hamiltonian, or simply an RK Hamiltonian.
The correspondence between RK Hamiltonians and Markov processes is illustrated in Fig.~\ref{fig: RK MC correspondence} and summarized in Table~\ref{tbl: RK MC correspondence}.
\begin{figure}[t]
    \centering
    \includegraphics[width=0.9\hsize]{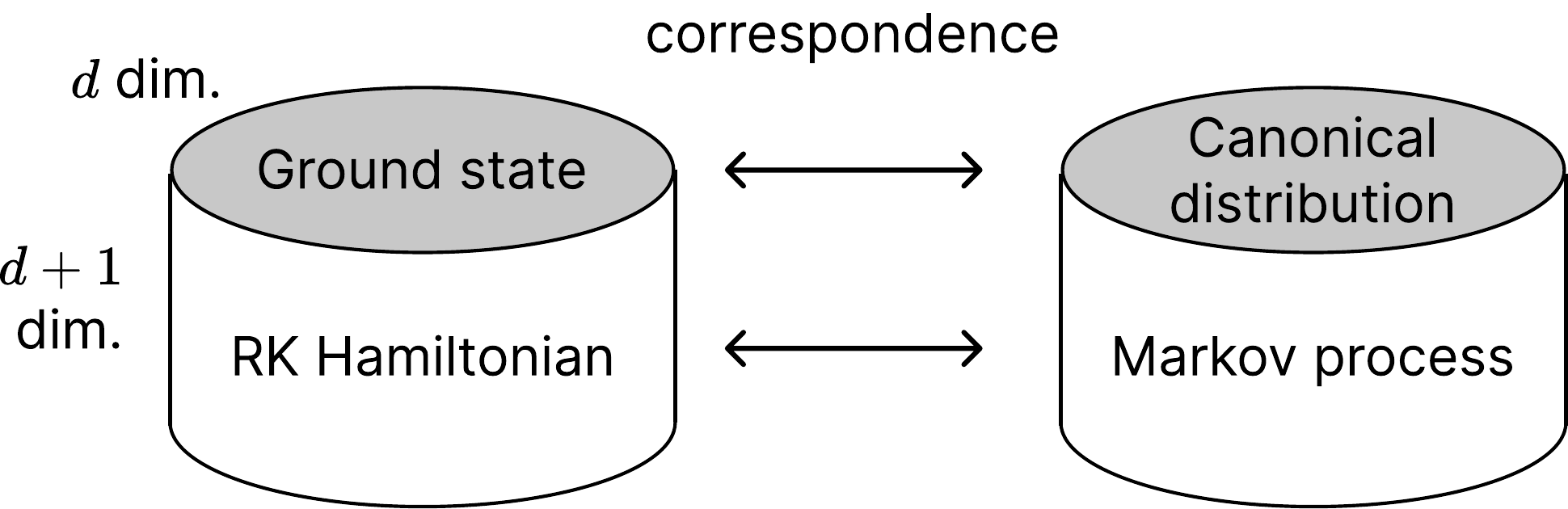}
    \caption{Correspondence between an RK Hamiltonian and a Markov process.
 The $d$-dimensional ground state corresponds to the $d$-dimensional Boltzmann distribution as in Eq.~\eqref{correspondence of expectation}.
 The Hamiltonian and the transition-rate matrix of a Markov process are connected by Eq.~\eqref{def of local transition rate}}.
    \label{fig: RK MC correspondence}
\end{figure}

\begin{table*}[t]
\centering
\begin{tabular}{cc}\hline\hline
RK Hamiltonians & Markov processes \\
\hline 
Hamiltonian  & Transition-rate matrix \\
$\hat H = \sum_i \hat H_i$ & $\hat W = \sum_i \hat W_i$\\
\hline
Imaginary-time Schr\"odinger equation & Master equation\\
$\mathrm{d}|\psi\rangle/\mathrm{d}t = -\hat H|\psi\rangle$ & $\mathrm{d}p(C)/\mathrm{d}t = \sum_{C'}(\hat W)_{CC'}p(C')$\\
\hline
Ground state  & Steady state (canonical distribution) \\
$|\Psi_\mathrm{RK}\rangle = \sum_{C} \sqrt{w(C)/Z} |C\rangle$  & $p_\mathrm{eq}(C) = w(C)/Z $ \\
\hline
Quantum expectation value & Classical expectation value \\
$\langle\Psi_\mathrm{RK}|\hat{\mathcal{O}}|\Psi_\mathrm{RK}\rangle$ & $\langle\mathcal{O}\rangle \coloneqq \sum_{C} \mathcal{O}(C) w(C)/Z$\\
\hline
Frustration-freeness & Conservation of probability \\
$\langle\Psi_\mathrm{RK}|\hat H_i = 0$ & $\sum_{C} (\hat W_i)_{CC'}=0$ \\
\hline
Nonpositive matrix elements & Non-negative matrix elements \\
$\langle C|\hat H_i|C'\rangle \le 0$ & $(\hat W_i)_{CC'} \ge 0$ \\
\hline
Symmetricity & Detailed balance condition \\
$\langle C | \hat{H}_i |C'\rangle = \langle C'| \hat{H}_i |C\rangle$ & $ (\hat W_i)_{CC'}w(C')  = (\hat W_i)_{C'C}w(C)$ \\\hline\hline
\end{tabular}
\caption{Correspondence between RK Hamiltonians and  Markov processes~\cite{henleyClassicalQuantumDynamics2004,castelnovoQuantumMechanicsClassical2005}.
Concrete correspondence is easily obtained from the similarity transformation given in Eq.~\eqref{def of local transition rate}   \label{tbl: RK MC correspondence}}
\end{table*}

The master equation in Eq.~\eqref{master eq} corresponds to the imaginary-time Schr\"odinger equation
\begin{align}
   &
 \frac{\mathrm{d}}{\mathrm{d}t}|\psi(t)\rangle = -\hat H|\psi(t)\rangle,\quad\\
   &
   \text{where}\quad|\psi(t)\rangle \coloneqq \hat S^{-1}\sum_{C \in \mathcal{C}}p(C,t)| C \rangle.
\end{align}
Since the steady state of the master equation is the canonical distribution $p_\mathrm{eq}(C) \coloneqq w(C)/Z$, the corresponding ground state of the RK Hamiltonian is given by
\begin{align}
 | \Psi_\mathrm{RK} \rangle = \hat S^{-1}\sum_{C \in \mathcal{C}}\frac{w(C)}{Z}| C \rangle
 = \sum_{C \in \mathcal{C}}\sqrt{\frac{w(C)}{Z}}| C \rangle.
   \label{Boltzmann state}
\end{align}
This state is properly normalized.
A notable property of
this ground state is that the expectation values in $| \Psi_\mathrm{RK} \rangle$ are equal to the classical one in Eq.~\eqref{expectation of Markov process} for a diagonal
operator $\hat{\mathcal{O}}$ that satisfies $\langle C | \hat{\mathcal{O}} | C' \rangle = \mathcal{O}(C)\delta_{CC'}$:
\begin{align}
 \langle\Psi_\mathrm{RK}|\hat{\mathcal{O}}|\Psi_\mathrm{RK}\rangle
 = \sum_{C} \mathcal{O}(C) \frac{w(C)}{Z} = \langle\mathcal{O}\rangle.
   \label{correspondence of expectation}
\end{align}
To guarantee the locality of an RK Hamiltonian, we assume that the transition rate $\hat W$ is decomposed as $\hat W = \sum_\mathrm{i \in \mathcal{S}} \hat W_i$, where sufficiently distant local terms $\hat W_i$ and $\hat W_j$ commute with each other.
Additionally, we assume that each local term is also a transition-rate matrix with the detailed balance condition:
\begin{align}&
 \langle C | \hat W_i | C' \rangle w(C') = \langle C' | \hat W_i | C \rangle w(C),\label{local DBC}\\
   &
 \sum_{C \in \mathcal{C}}\langle C | \hat W_i | C' \rangle = 0,
   \label{local prob conservation}\\
   &
 \langle C | \hat W_i | C' \rangle \ge 0\quad(C \ne C').
   \label{local matrix elements positivity}
\end{align}
The Markov processes satisfying these properties include the Gibbs sampling (or the heat-bath) algorithm~\cite{gemanStochasticRelaxationGibbs1984} and the Metropolis-Hastings algorithms~\cite{metropolisEquationStateCalculations1953,hastingsMonteCarloSampling1970}.
These algorithms and corresponding RK Hamiltonians are discussed in Appendix~\ref{sec: MCMC}.

The local terms of the corresponding RK Hamiltonian are given by
\begin{align}&
    \hat W_i \coloneqq -\hat S\hat H_i\hat S^{-1}.
    \label{def of local transition rate}
\end{align}
In the following discussion, we show that the RK Hamiltonian $\hat H = \sum_{i \in \mathcal{S}}\hat H_i$ is FF by this decomposition if conditions in Eqs.~\eqref{local DBC}-\eqref{local matrix elements positivity} hold.
First, we show that the operator $\hat H_i$ is Hermitian.
Indeed, $\hat H_i$ becomes real symmetric.
This follows from
\begin{align}
 \langle C |\hat H_i| C' \rangle &
 = - \frac1{\sqrt{w(C)}} \langle C|\hat W_i|C'\rangle\sqrt{w(C')} \nonumber\\
   &
 = - \frac1{\sqrt{w(C')}} \langle C'|\hat W_i|C\rangle\sqrt{w(C)} \nonumber\\
   &
 = \langle C' |\hat H_i| C \rangle.
   \label{symmetricity}
\end{align}
In the first line, we used $\hat H_i = -\hat S^{-1}\hat W_i\hat S$ and the definition of $\hat S$ in Eq.~\eqref{def of S}.
In the second line, we used the detailed balance condition in Eq.~\eqref{local DBC}.
Since the matrix elements of a local term $\hat H_i$ are real by construction, $\hat H_i$ is Hermitian.

The probability conservation for the local term in
Eq.~\eqref{local prob conservation} implies that
\begin{align}
 \langle \Psi_\mathrm{RK} |\hat H_i &
 = -\langle \Psi_\mathrm{RK} |\hat S^{-1}\hat W_i\hat S \nonumber\\
   &
 = -\sum_{C,C' \in \mathcal{C}}\langle C | \hat W_i | C' \rangle\langle C' |\hat S \nonumber\\
   &
 = 0.
   \label{GS annihilation}
\end{align}
Therefore, the operator $\hat H_i$ annihilates the state $| \Psi_\mathrm{RK} \rangle$.
Since $\hat H_i$ is positive semidefinite~\footnote{
 Since $(H_i)_{CC'}$ is real symmetric from Eq.~\eqref{symmetricity}, eigenstates of the matrix $H_i$ can be chosen to be real.
 Let $v$ be an eigenstate of $H_i$ with the lowest eigenvalue $\mu_i$ and $u$ be a vector whose components are given by $u_C = |v_C|$.
 Since off-diagonal elements of $H_i$ are non-positive, we have $\mu_i = \sum_{C,C'}v_C(H_i)_{CC'}v_{C'} \ge u_C(H_i)_{CC'}u_{C'}$.
 Then, the variational principle implies that $u$ is also an eigenvector of $H_i$ with the eigenvalue $\mu_i$.
On the other hand, $\sqrt{w(C)}$ is an eigenvector of $H_i$ with the eigenvalue $0$ and $\sum_{C}\sqrt{w(C)}u_C = \sum_{C}\sqrt{w(C)}|v_C| \ne 0$.
Therefore, $\mu_i=0$.
}, $\hat H_i| \Psi_\mathrm{RK} \rangle = 0$ means that $\hat H_i$ is minimized in the state $| \Psi_\mathrm{RK} \rangle$.
Therefore, the RK Hamiltonian defined by Eq.~\eqref{def of local transition rate} is FF if conditions in
Eqs.~\eqref{local DBC}-\eqref{local matrix elements positivity} hold.

\subsubsection{Critical slowing down and dynamical exponent}

The imaginary time evolution operator ${{e}}^{-\hat Ht}$ associated with an RK Hamiltonian corresponds to the transition matrix ${{e}}^{\hat Wt}$ of a classical stochastic process.
Since $-\hat H$ and $\hat W$ are connected by a similarity transformation, their spectra are identical.
For a diagonal operator $\mathcal{O}$ in the basis $\{|C\rangle\}$, the autocorrelation function is defined by
\begin{align}
A_{\mathcal{O}}(t)
\coloneqq \frac{\langle\Psi_\mathrm{RK}|\hat{\mathcal{O}}({{e}}^{-\hat{H} t}-\hat G)\hat{\mathcal{O}}|\Psi_\mathrm{RK}\rangle}{\langle\Psi_\mathrm{RK}|\hat{\mathcal{O}}(\hat{\mathbbm{1}}-\hat G)\hat{\mathcal{O}}|\Psi_\mathrm{RK}\rangle}.
\end{align}
If the ground state is unique, it coincides with the classical expression~\cite{landauGuideMonteCarlo2014}
\begin{align}
 A_{\mathcal{O}}(t) = \frac{\langle\mathcal{O}(t)\mathcal{O}(0)\rangle-\langle\mathcal{O}\rangle^2}{\langle\mathcal{O}^2\rangle-\langle\mathcal{O}\rangle^2}.
\end{align}
The spectral decomposition of ${{e}}^{-\hat Ht}$is given by
\begin{align}
 {{e}}^{-\hat H t} = \hat G + \sum_{\alpha} {{e}}^{-E_\alpha t} |\alpha\rangle\langle\alpha|,
\end{align}
where $| \alpha \rangle$ are the excited eigenstates of $\hat H$ and $E_\alpha$ are the corresponding energy eigenvalues.
Because of this decomposition, the autocorrelation function $A_{\mathcal{O}}(t)$ decays as ${{e}}^{-E_1t}$ in the long-time limit.
We then define the relaxation time $\tau$ by
\begin{align}
 \tau \coloneqq \frac1{E_1} = \frac1{\gap(\hat H)}.
\end{align}
The vanishing of the spectral gap implies a divergence of the relaxation time, a phenomenon known as critical slowing down.
The dynamical exponent $z$ for a classical stochastic system is defined using the relation
\begin{align}
 \tau \sim L^z,
\end{align}
where $L$ denotes the linear system size.
This definition of $z$ is identical to the quantum one.

Phase transitions in classical statistical systems, driven by changes in their Boltzmann weights, can be mapped to quantum phase transitions in the corresponding RK Hamiltonians.
At these quantum critical points, which are also known as conformal quantum critical points, ground-state correlation functions typically exhibit power-law decay.
Correspondingly, the associated Markovian dynamics manifest critical slowing down.
It is noteworthy that these conformal quantum critical points are distinct from many other quantum critical points whose low-energy dynamics are described by conformal field theories.
While equal-time correlation functions decay algebraically in both scenarios, their underlying dynamical mechanisms and properties can differ.

\subsubsection{Markov chain Monte Carlo methods}
We briefly explain the implication of the dynamical exponent in the Markov chain Monte Carlo methods.
Let us consider the calculation of the expectation value of a physical quantity $\mathcal{O}$ given by
\begin{align}
    \langle \mathcal{O} \rangle = \sum_{C} \mathcal{O}(C)\frac{w(C)}{Z},
\end{align}
where $w(C)/Z$ is the canonical distribution.
Since directly performing this summation over all possible configurations $C$ is computationally challenging for systems with many degrees of freedom, efficient sampling methods are necessary.
One of the most effective and widely used sampling techniques is the Markov chain Monte Carlo method~\cite{landauGuideMonteCarlo2014}.

In MCMC methods, we simulate a Markov chain whose stationary distribution is the canonical distribution $p_\mathrm{eq}(C) = w(C)/Z$.
We start the simulation with some initial configuration $C(0) = C_0$.
At each step, we randomly generate a new configuration $C(t+\delta t)$ from $C(t)$ using the master equation in Eq.~\eqref{master eq}.
By iterating this procedure over sufficiently long times, the distribution of configurations visited by the Markov chain converges to the target canonical distribution.

At a critical point, the relaxation time $\tau$ diverges as $\tau \sim L^z$, inducing a critical slowdown.
This significantly increases the cost of MCMC calculations.
The statistical error of expectation values is proportional to $1/\sqrt N$ , where $N$ is the effective number of independent samples. As two configurations are effectively independent after time $\tau$, we have
\begin{align}
    \text{error} \propto \frac1{\sqrt N} \propto \sqrt{\frac{\tau}{T}} \propto \frac{L^{z/2}}{\sqrt{T}},
\end{align}
where $T$ is the total time of the simulation.
Therefore, reducing the dynamical exponent $z$ is crucial for improving the efficiency of MCMC methods at critical points.

\subsection{Proof of $z \ge 2$ for RK Hamiltonians}

Let us now prove the main result of this section: the lower bound $z \ge 2$ for RK Hamiltonians of critical points.
This bound is stronger than the previous rigorous bound $\Delta \ge \gamma$~\cite{abeDynamicsIsingModel1969,halperinRigorousInequalitiesSpinRelaxation1973,lubetzkyCriticalIsingSquare2012}, which is equivalent to $z \ge 2-\eta$ where $\eta$ is the anomalous dimension~\footnote{
 The exponent $\Delta$ is defined by $\tau \sim |T-T_c|^{-\Delta}$, where $T$ is the temperature and $T_c$ is the critical temperature.
 The dynamical exponent $z$ is written as $z = \Delta/\nu$, where $\nu$ is the correlation length exponent.
 Using the Fisher's relation $\gamma/\nu = 2-\eta$, the bound $\Delta \ge \gamma$ gives $z \ge 2-\eta$.
 }.
The proof of $z \ge 2-\eta$ using a variational state is given in Appendix~\ref{sec: variational bound}.

\subsubsection{$z \ge 2$ for RK Hamiltonians of critical points}

Let us start with the definition of ergodicity in Markov processes.
Consider a graph $\mathcal{G}$ where each vertex corresponds to a classical configuration $C$. An edge connects two configurations, $C$ and $C'$, if the matrix element $\langle C'|\hat{W}|C\rangle$ is positive.
When this graph $\mathcal{G}$ is connected, the associated Markov process and the corresponding RK Hamiltonian are termed ergodic.
For such ergodic systems, the application of the Perron-Frobenius theorem to the operator $e^{-\hat{H} t}$ establishes the uniqueness of its ground state. 
The subsequent discussion concentrate on ergodic RK Hamiltonians, while nonergodic RK Hamiltonians are addressed in Appendix~\ref{sec: non-ergodic RK}.

\begin{thm}
 \label{thm: RK}
 For ergodic RK Hamiltonians of critical points, the dynamical exponent satisfies the lower bound $z \ge 2$.
\end{thm}
\begin{proof}
From the assumption of uniqueness of the ground state, the state $|\Psi_\mathrm{RK}\rangle$ given in Eq.~\eqref{Boltzmann state} is the unique ground state of the RK Hamiltonian.
Since the corresponding classical statistical system is at a critical point, there exists a local quantity $\mathcal{O}_{\bm{x}}$ such that
\begin{align}
\frac{|\langle\mathcal{O}_{\bm{x}} \mathcal{O}_{\bm{y}}\rangle - \langle\mathcal{O}_{\bm{x}}\rangle\langle\mathcal{O}_{\bm{y}}\rangle|}{\sqrt{\langle\mathcal{O}_{\bm{x}}^2\rangle\langle\mathcal{O}_{\bm{y}}^2\rangle}} \ge \mathrm{const}\times |\bm{x}-\bm{y}|^{-2\Delta_{\mathcal{O}}},
\end{align}
where $\Delta_{\mathcal{O}}$ is the scaling dimension of $\mathcal{O}_{\bm{x}}$.
Let $|\bm{x}-\bm{y}| = \mathrm{const}\times L$ and $\hat{\mathcal{O}}_{\bm{x}} \coloneqq \sum_{C} \mathcal{O}_{\bm{x}}(C) |C\rangle\langle C|$.
From the correspondence between classical and quantum correlation functions in Eq.~\eqref{correspondence of expectation}, we have
\begin{align}
    \frac{\langle\Psi_\mathrm{RK}|\hat{\mathcal{O}}_{\bm{x}}(\hat{\mathbbm{1}} - |\Psi_\mathrm{RK}\rangle\langle\Psi_\mathrm{RK}|)\hat{\mathcal{O}}_{\bm{y}}|\Psi_\mathrm{RK}\rangle}{\|\hat{\mathcal{O}}^\dagger_{\bm{x}}|\Psi_\mathrm{RK}\rangle\|\|\hat{\mathcal{O}}_{\bm{y}}|\Psi_\mathrm{RK}\rangle\|} \approx \mathrm{const}\times L^{-2\Delta_\mathcal{O}}.
\end{align}
Thus, we can use Theorem~\ref{thm: main} to obtain $z \ge 2$.
\end{proof}
Correspondingly, the following theorem for Markov processes also holds.
\begin{thm}
 For the ergodic Markov processes that relax to the canonical distribution of a critical point, the dynamical exponent $z$ satisfies the lower bound $z \ge 2$ if (i) the transition-rate matrix $\hat W$ is decomposed as
 \begin{align}
    \hat W = \sum_i \hat W_i,
 \end{align}
 where each term is a transition-rate matrix with finite range support, and (ii) the detailed balance condition holds for each local term $\hat W_i$.~\footnote{
 This condition is equivalent to the detailed balance condition for $\hat W$ if each off-diagonal element $W_{CC'}$ is attributed to a single local term $\hat W_i$.
 }
\end{thm}

In the above discussion, we assumed the existence of a power-law correlation function at the critical point. A more mathematically rigorous discussion in the case of the Ising model is given in Ref.~\cite{masaokaRigorousLowerBound2025}.

\subsection{Implications of locality and detailed balance}
This theorem serves as a no-go theorem that prohibits dynamics faster than $z = 2$ for dynamics satisfying locality and detailed balance. This class of MCMC methods includes the well-known Metropolis-Hastings algorithm~\cite{metropolisEquationStateCalculations1953,hastingsMonteCarloSampling1970} and the Gibbs sampling (heat-bath) algorithm~\cite{gemanStochasticRelaxationGibbs1984}.
Numerical values of dynamical exponents of such dynamics for various critical points are summarized in Table~\ref{table of dynamical exponent} in Sec.~\ref{sec: introduction}. The dynamical exponents are close to the optimal value $z=2$ of this class, meaning they have little room for improvement.

If we relax either of these conditions, it becomes possible to obtain dynamics with $z < 2$.
Let us focus on two particular examples summarized in Table~\ref{tbl: faster dynamics}.
The asymmetric simple exclusion process (ASEP), which is described by the transition-rate
\begin{align}
W = -\sum_{i=1}^L \Bigg[
   &
   \frac{1}{4}(1-\hat{\sigma}^z_i\hat{\sigma}^z_{i+1})
   -\frac{1+s}{2}\hat{\sigma}^+_i\hat{\sigma}^-_{i+1}\nonumber\\
   &
   -\frac{1-s}{2}\hat{\sigma}^-_i\hat{\sigma}^+_{i+1}
   +\frac{s}{2}(\hat{\sigma}^z_i\hat{\sigma}^z_{i+1})
 \Bigg].
\end{align}
This is a local model, but it does not satisfy the detailed balance condition.
Under the periodic boundary condition, it exhibits a dynamical exponent of $z = 3/2 < 2$~\cite{gwaBetheSolutionDynamicalscaling1992,kimBetheAnsatzSolution1995}.
The corresponding Hamiltonian $H = -W$, obtained via Eq.~\eqref{RK-MC map}, is real and non-symmetric, hence non-Hermitian.

Another example is the Wolff cluster algorithm~\cite{wolffCollectiveMonteCarlo1989}, which satisfies the detailed balance condition but updates clusters of spins collectively in a single step.
Since the cluster size can vary arbitrarily, this process is inherently nonlocal.
In the case of the Wolff algorithm for the 2D critical Ising model, numerical studies have estimated that $z \approx 0.3 < 2$~\cite{liuDynamicScalingClassical2014}.
Other sophisticated MCMC methods also achieve faster relaxation by employing nonlocal time evolution or by breaking the detailed balance condition~\cite{swendsenNonuniversalCriticalDynamics1987, suwaMarkovChainMonte2010,turitsynIrreversibleMonteCarlo2011}, as they can go around the no-go theorem.

\begin{table}[H]
    \centering
        \begin{tabular}{cccc}\hline\hline
 Model & Detailed balance & Locality & $z$ \\
        \hline
 ASEP & $\times$ & $\checkmark$ & $3/2$ \cite{gwaBetheSolutionDynamicalscaling1992,kimBetheAnsatzSolution1995} \\
 Wolff algorithm & $\checkmark$ & $\times$ & $ 0.3$ \cite{wolffCollectiveMonteCarlo1989}\\\hline\hline
    \end{tabular}
    \caption{Examples of faster dynamics than $z = 2$.}
    \label{tbl: faster dynamics}
 \end{table}

We also note that the presence of conserved charges or constraints on stochastic dynamics can alter the dynamic universality classes.
In general, such conserved quantities or dynamical constraints tend to increase the dynamical exponent $z$ or transform a gapped model into a gapless one by slowing down the relaxation process.
For instance, the Kawasaki dynamics~\cite{kawasakiDiffusionConstantsCritical1966} of the critical Ising model with conserved magnetization is described by the model $B$ in the Hohenberg-Halperin classification~\cite{hohenbergTheoryDynamicCritical1977} and has a larger dynamical exponent than the model $A$ without conserved charges.
While our bound, $z \ge 2$, holds regardless of the presence of conserved charges or constraints, it does not capture the changes in $z$ due to their presence.

It is also worth noting a field-theoretic implication of our result.
The effective theories of RK Hamiltonians are given by stochastic quantization~\cite{parisiPERTURBATIONTHEORYGAUGE1981,dijkgraafRelatingFieldTheories2010}. Consequently, our Theorem~\ref{thm: RK} leads to the conclusion that stochastic quantizations of conformal field theories satisfy $z\ge2$. However, a field-theoretic justification of this statement in the absence of the underlying lattice remains an open question.

\subsubsection{Generalization of RK Hamiltonians}
We can consider several simple extensions to the above results for the RK Hamiltonians.
While the correspondence with Markov processes may not hold in these extensions, the behavior of the dynamical exponent remains the same.

One such extension involves deforming the local terms of the Hamiltonian without altering its ground-state subspace. This modification leaves the ground states and the dynamical exponent unchanged, although it generally breaks the correspondence with classical systems.

Another possible extension is to introduce a configuration-dependent phase factor into the RK ground state of Eq.~\eqref{Boltzmann state}:
\begin{align}
 |\Psi\rangle \coloneqq \sum_{C} {{e}}^{i\theta(C)} \sqrt{\frac{w(C)}{Z}} |C\rangle,
    \label{RK GS with phase}
\end{align}
where, the phase $\theta(C)$ is a real function.
This state maintains the correspondence between quantum and classical expectation values, as described in Eq.~\eqref{correspondence of expectation}.
This implies that the proof of Theorem~\ref{thm: RK} is still valid for this generalization.
Further generalizations of the ground state using projected entangled pair states are discussed in Appendix~\ref{sec: PEPS}.

\subsection{Example: kinetic Ising model}\label{sec: examples}
As an example of the RK Hamiltonian, let us consider the RK Hamiltonian of the Ising model on any lattice $\Lambda$.
The resulting Hamiltonian and the corresponding Markov process are known as the kinetic Ising model.

In the Ising model, classical configurations $C \in \mathcal{C}$ are labeled by the list of the values of spins $s_i=\pm1$ as $C=\{s_i\}_{i \in \Lambda}$
The Boltzmann weight of the Ising model is
\begin{align}
 w(C) = \exp\left( K\sum_{\langle i,j\rangle}s_i s_j\right),
\end{align}
where $\langle i,j \rangle$ is the edge connecting $i$ and $j$, and $K$ is the coupling constant.
The unique ground state $|\Psi_\mathrm{RK}\rangle$ corresponding to this Boltzmann weight is given by
\begin{align}
 |\Psi_\mathrm{RK}\rangle = \frac{1}{\sqrt Z}\sum_{\{s\}}\exp\left(\frac{K}{2}\sum_{\langle i,j\rangle}s_{i}s_{j}\right)| \{s\} \rangle, \label{RK Ising state}
\end{align}
where $\{s\}$ represents all possible spin configurations and $Z$ is the partition function.

Let $B_i$ be the set of lattice sites adjacent to $i$, and let $E_i = \Lambda\setminus(\{i\} \cup B_i)$ be the remaining sites.
The transition rate of the RK Hamiltonian of the Gibbs sampling algorithm for the Ising model with no additional constraints is given by
\begin{align}
   \langle \{s\} | \hat W_i | \{s'\} \rangle
 = \left( \frac{\exp(Ks_i\sum_{j \in B_i}s_j)}{2\cosh(K\sum_{j \in B_i}s_j)} - \delta_{s_is'_i} \right) \prod_{j \in B_i}\delta_{s_js'_j}.
\end{align}
The corresponding Hamiltonian is computed as
\begin{align}
 \hat H_i = (\hat{\mathbbm{1}}_{i} \otimes \hat{\mathbbm{1}}_{B_i} - |\Psi_i\rangle\langle\Psi_i|)\otimes\hat{\mathbbm{1}}_{E_i},
\end{align}
where 
\begin{align}&
 |\Psi_i\rangle \coloneqq \sum_{s_i, \{s_j\}_{j \in B_i}} \frac{\exp(\frac{K}{2}s_i\sum_{j \in B_i} s_j)|s_i, \{s_j\}_{j \in B_i}\rangle}{\sqrt{2\cosh(K\sum_{j \in B_i}s_j)}}.
\end{align}
This can be described using the Pauli operators as
\begin{align}
\hat H_i = \frac{1}{2\cosh(K\sum_{j \in B_i}\hat\sigma^z_j)}\left( {{e}}^{K\hat\sigma^z_i \sum_{j \in B_i}\hat\sigma^z_j} - \hat\sigma^x_i \right),
\end{align}
where $\hat\sigma^z_i$ and $\hat\sigma^x_i$ are defined so that
\begin{align}
   \hat\sigma^z_i| s_i \rangle = s_i| s_i \rangle,\quad
   \hat\sigma^x_i| s_i \rangle = | -s_i \rangle.
\end{align}
This Hamiltonian is ergodic since one can obtain any spin configuration from a particular spin configuration by repeating local spin flips.

In $d \ge 2$-dimensional lattice, we obtain a critical point by tuning the coupling constant $K$ to the critical value $K_c$. 
The dynamical exponent $z$ satisfies $z \ge 2$ due to ergodicity and Theorem~\ref{thm: RK}.
For the square lattice, the exact value of $K_c$ is $\frac{1}{2}\ln(1+\sqrt{2}) \approx 0.44$, and the numerical value of $z = 2.1667(5)$ is computed by the MCMC method~\cite{nightingaleMonteCarloComputation2000}.
A typical time development of the MCMC simulation is depicted in Fig.~\ref{fig: kinetic Ising MCMC}.
For the cubic lattice, numerical values of $K_c \approx 0.22$~\cite{liuCriticalTemperatureTwodimensional2003} and $z = 2.0245(15)$ are known~\cite{hasenbuschDynamicCriticalExponent2020}.

\begin{figure*}[t]
   \centering
   \begin{minipage}[t]{0.49\hsize}
      \centering
      \includegraphics[width=\hsize]{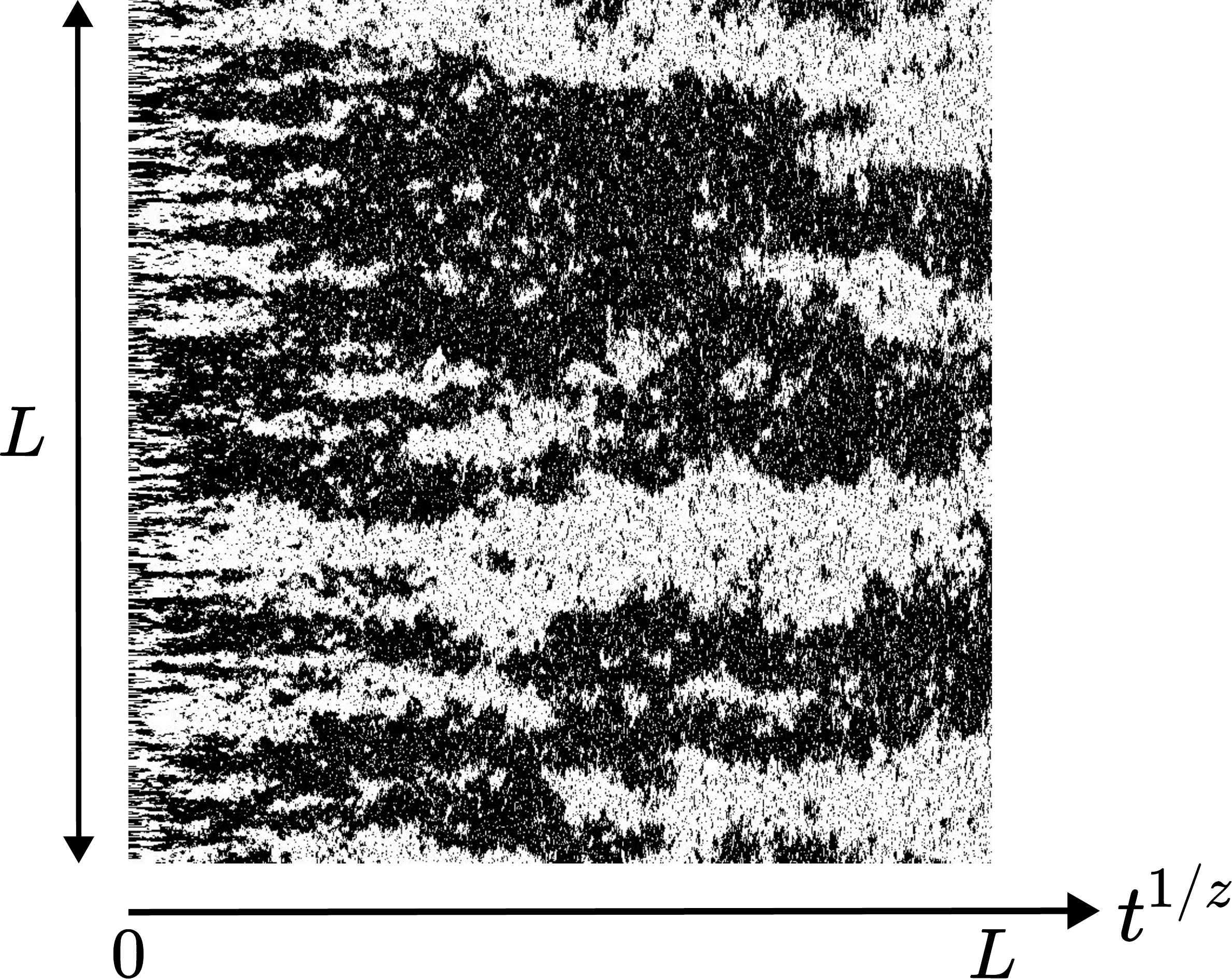}
      \subcaption{}
   \end{minipage}
   \begin{minipage}[t]{0.49\hsize}
      \centering
      \includegraphics[width=0.9\hsize]{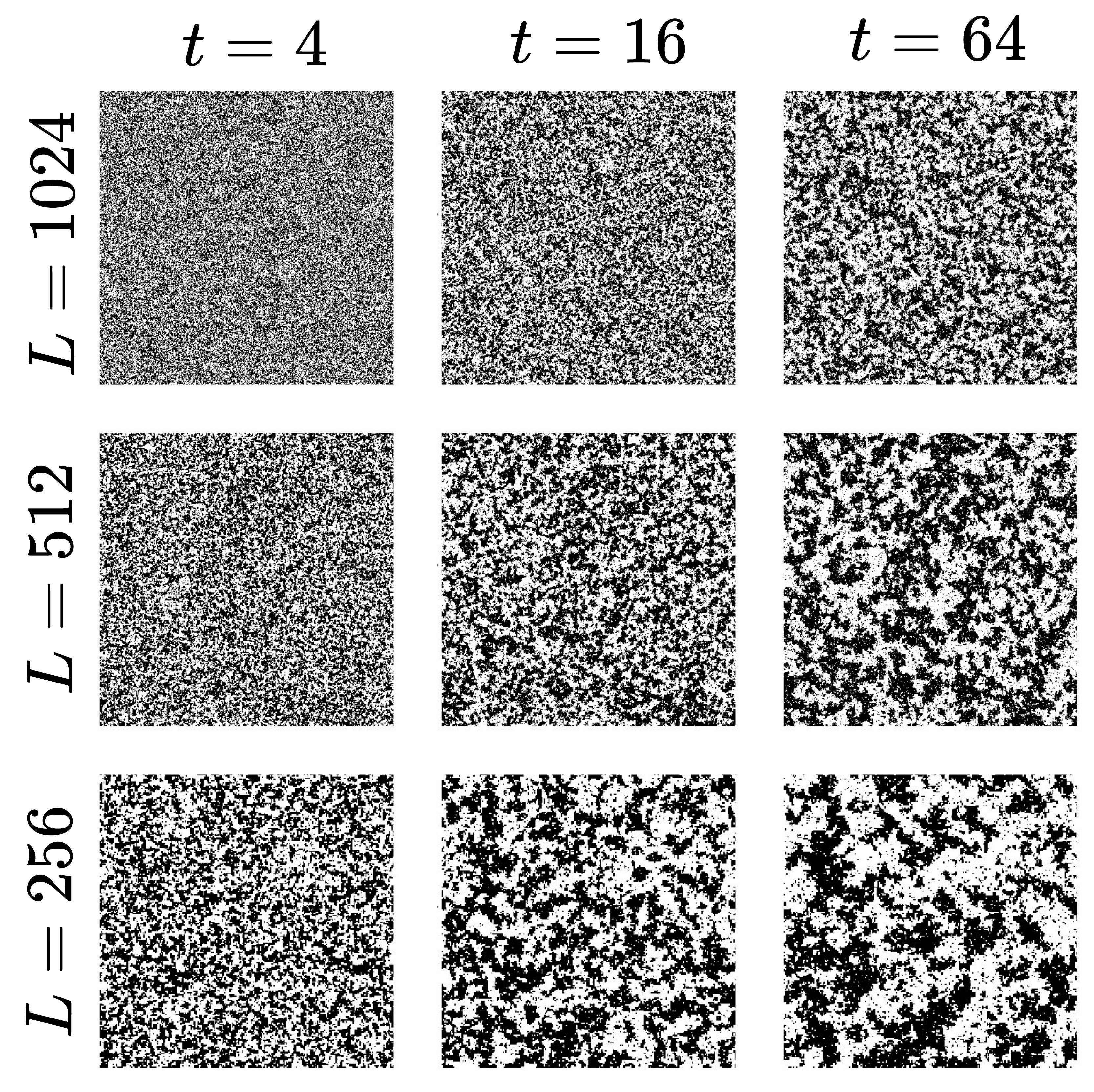}
      \subcaption{}
   \end{minipage}
   \caption{
       MCMC simulation of the two-dimensional kinetic Ising model at criticality.
 Black (white) pixels denote up (down) spins.
 (a) A $(1+1)$-dimensional slice of the simulation. The vertical axis represents a one-dimensional spatial slice, and the horizontal axis the rescaled time $t^{1/z}$ with $z \approx 2.17$.
 The characteristic time of relaxation scales with system size as $L^z$.
 (b)
 Snapshots of the spin configuration for various system sizes and times.
 Rows correspond to system sizes $L=1024$, $512$, and $256$ (top to bottom) and columns show times $t=4$, $16$, and $64$ (left to right).
  }
      \label{fig: kinetic Ising MCMC}
\end{figure*}

\section{Extension of the Gosset-Huang inequality}
\label{sec: hidden}

In this section, we discuss an extension of the Gosset-Huang inequality which broadens its applicability beyond correlations generated by strictly local operators.
In this extension, the excitation $\hat{\mathcal{O}}| \Psi \rangle$ created by a local operator $\hat{\mathcal{O}}$ from the ground state $| \Psi \rangle$ is replaced by a general excitation $| \mathcal{O} \rangle$ with a localized energy density.
This generalization allows us to prove that $z \ge 2$ for models whose ground states do not exhibit power-law decaying correlation functions for any local operators but possess hidden power-law correlations.

\subsection{Extended Gosset-Huang inequality}
To formulate an extended Gosset-Huang inequality, we first define the distance $D(|\mathcal{O}\rangle, |\mathcal{O}'\rangle)$ between general states $|\mathcal{O}\rangle$ and $|\mathcal{O}'\rangle$.
We define the energy support $\mathcal{S}(|\mathcal{O}\rangle)$ for a state $|\mathcal{O}\rangle$ by
\begin{align}
 \mathcal{S}(|\mathcal{O}\rangle) \coloneqq \{i \in \mathcal{S} \mid \hat H_i|\mathcal{O}\rangle \ne 0 \},
\end{align}
where $\mathcal{S}$ is the index set for the local terms $\{\hat H_i\}$.
Note that $\mathcal{S}(|\Psi\rangle)$ is empty for a ground state $|\Psi\rangle$, and $\mathcal{S}(|\mathcal{O}\rangle)$ is the whole index set $\mathcal{S}$ for almost all states.
Previously, we considered excitations of the form
\begin{align}
 |\mathcal{O}\rangle = \hat{\mathcal{O}}|\Psi\rangle,
\end{align}
where $\hat{\mathcal{O}}$ is a local operator.
In this case, $\mathcal{S}(|\mathcal{O}\rangle)$ consists of indices around the position of $\hat{\mathcal{O}}$.
However, we can also consider excitations with local energy supports but that are generated by nonlocal operators.
Furthermore, we can consider not only pointlike excitations, but also higher-dimensional excitations.
The distance $D(|\mathcal{O}\rangle, |\mathcal{O}'\rangle)$ between two states $|\mathcal{O}\rangle$ and $|\mathcal{O}'\rangle$ is given by
\begin{align}
 D(|\mathcal{O}\rangle, |\mathcal{O}'\rangle)
    &
 \coloneqq \min \{d(i, j) \mid i \in \mathcal{S}(|\mathcal{O}\rangle), j \in \mathcal{S}(|\mathcal{O}'\rangle) \} \nonumber\\
    &
 = \min \{d(i, j) \mid \hat H_i|\mathcal{O}\rangle \ne 0, \hat H_j|\mathcal{O}'\rangle \ne 0 \},
\end{align}
where $d(i,j)$ is the distance measured in the interaction graph, which is defined in Sec.~\ref{sec: GH inequality}.

The correlation function between two states $| \mathcal{O} \rangle$ and $| \mathcal{O}' \rangle$ is written as
\begin{align}
 \frac{|\langle\mathcal{O}|(\hat{\mathbbm{1}}-\hat{G})|\mathcal{O}'\rangle|}{\||\mathcal{O}\rangle\| \||\mathcal{O}'\rangle\|},
\end{align}
where $\hat G$ denotes the projector onto the ground space.
Then, an extended Gosset-Huang inequality is given by
\begin{align}&
 \frac{|\langle\mathcal{O}|(\hat{\mathbbm{1}}-\hat{G})|\mathcal{O}'\rangle|}{\||\mathcal{O}\rangle\| \||\mathcal{O}'\rangle\|} \nonumber\\
   &
 \le 2\exp\left[ -\left( \frac{D(| \mathcal{O} \rangle, | \mathcal{O}' \rangle)-1}{c-1} -2 \right)\sqrt{\frac{\epsilon}{g^2+\epsilon}}\right].
\end{align}
Using this inequality, we obtain the following extension of Theorem~\ref{thm: main}.
\begin{thm}\label{thm: general}
 For a frustration-free Hamiltonian, the dynamical exponent $z$ satisfies $z \ge 2$ if there exist two states $|\mathcal{O}\rangle$ and $|\mathcal{O}'\rangle$ such that $D(|\mathcal{O}\rangle, |\mathcal{O}'\rangle) \ge AL~(A>0)$ and
    \begin{align}
 \frac{|\langle\mathcal{O}|(\hat{\mathbbm{1}}-\hat{G})|\mathcal{O}'\rangle|}{\||\mathcal{O}\rangle\| \||\mathcal{O}'\rangle\|} \ge BL^{-\alpha},
        \label{assumption of general thm}
    \end{align}
 where $\hat G$ denotes the projector onto the ground space of $\hat H$ and $\alpha \ge 0$ and $B > 0$ are constants independent of $L$.
\end{thm}

\begin{proof}
 The proof is identical to that of Theorem~\ref{thm: main}.
\end{proof}

\subsection{Examples of hidden correlations}

We now present simple examples of models whose ground states exhibit no power-law correlations between local operators but possess hidden power-law correlations in the extended sense.

\subsubsection{1+1D zero-temperature kinetic Ising model}
Let us consider the 1+1D kinetic Ising model under the periodic boundary condition:
\begin{align}&
    \hat H = \sum_{i=1}^L \hat H_i, \\
    &
    \hat H_i = \frac{{{e}}^{-K\hat\sigma^z_i(\hat\sigma^z_{i-1}+\hat\sigma^z_{i+1})}-\hat\sigma^x_i}{2\cosh[K(\hat\sigma^z_{i-1}+\hat\sigma^z_{i+1})]}.
\end{align}
A short computation yields
\begin{align}
    \hat H_i &
 = \frac1{2}\hat{\mathbbm{1}}-\frac{J}{4}\hat\sigma^z_i(\hat\sigma^z_{i-1}+\hat\sigma^z_{i+1}) 
    -\frac h{4}\hat\sigma^x_i + \frac{J'}{4}\hat\sigma^z_{i-1}\hat\sigma^x_i\hat\sigma^z_{i+1},
\end{align}
where 
\begin{align}
 J = \tanh(2K),\quad h = \frac{2\cosh^2K}{\cosh(2K)},\quad J' = \frac{2\sinh^2K}{\cosh(2K)}.
\end{align}
The 1D Ising model has a critical point at $K \to \infty$.
Correspondingly, the 1+1D kinetic Ising model becomes gapless at $K \to \infty$.
In the zero-temperature limit $K \to \infty$, the Hamiltonian becomes
\begin{align}
    \hat H_i = \frac1{2}\left(\hat{\mathbbm{1}}-\frac1{2}\hat\sigma^z_i(\hat\sigma^z_{i-1}+\hat\sigma^z_{i+1})-\frac1{2}\hat\sigma^x_i(\hat{\mathbbm{1}}-\hat\sigma^z_{i-1}\hat\sigma^z_{i+1})\right).
\end{align}
The kernel of this local term is
\begin{align}
    \ker \hat H_i = \Span\{|000\rangle, |111\rangle, |1{+}0\rangle, |0{+}1\rangle\},
\end{align}
where $|{+}\rangle \coloneqq (|0\rangle+|1\rangle)/\sqrt2$.
The ground state under the periodic boundary condition is given by
\begin{align}
 |\Psi_0\rangle \coloneqq |0\cdots0\rangle,\quad|\Psi_1\rangle \coloneqq |1\cdots1\rangle.
\end{align}
While these ground states seem to have no correlation, they have hidden correlations due to the nontrivial kernel of the local terms of the Hamiltonian.
We can show $z \ge 2$ using such correlation functions.
Let us consider the following ``local'' excitation states constructed from superpositions of two-domain wall configurations.
\begin{align}&
 |\mathcal{O}^+_{1/2}\rangle \coloneqq |10\cdots0\rangle+|110\cdots0\rangle+\cdots+|1\cdots10\rangle, \\
    &
 |\mathcal{O}^-_{1/2}\rangle \coloneqq |01\cdots1\rangle+|001\cdots1\rangle+\cdots+|0\cdots01\rangle.
\end{align}
These states are the ground states of the open-boundary Hamiltonian $\hat H = \sum_{i=2}^{L-1}\hat H_i$.
We also define the excitations for other positions as $|\mathcal{O}^\pm_{i+(1/2)}\rangle \coloneqq \hat T^i|\mathcal{O}^\pm_{(1/2)}\rangle$, where $\hat T$ is the translation operator.
These states cannot be generated by applying local operators to the ground states.
However, these states can be treated as local excitations since
\begin{align}
    \hat H_j|\mathcal{O}^\pm_{i+(1/2)}\rangle = 0 \quad (j \ne i, i+1),
\end{align} 
and thus the energy support for $|\mathcal{O}^\pm_{i+(1/2)}\rangle$ is given by a localized region $\{i, i+1\}$.
The distance between the two states $|\mathcal{O}^+_{i+(1/2)}\rangle, |\mathcal{O}^-_{j+(1/2)}\rangle$ is computed as
\begin{align}
 D(|\mathcal{O}^+_{i+(1/2)}\rangle, |\mathcal{O}^-_{j+(1/2)}\rangle) = |i-j|-1.
\end{align}
We can also explicitly write down the operators that create these excitations:
\begin{align}
    \hat{\mathcal{O}}^+_{i+(1/2)} &= \hat\sigma^x_{i+1}\hat\sigma^z_{i+2}+\hat\sigma^x_{i+1}\hat\sigma^x_{i+2}\hat\sigma^z_{i+3}+\cdots \nonumber\\
    &\quad\quad\quad\quad\quad
 \cdots+\hat\sigma^x_{i+1}\hat\sigma^x_{i+2}\cdots\hat\sigma^x_{i-1}\hat\sigma^z_i,
    \\
    \hat{\mathcal{O}}^-_{i+(1/2)} &= \hat\sigma^x_i\hat\sigma^z_{i-1}+\hat\sigma^x_i\hat\sigma^x_{i-1}\hat\sigma^z_{i-2}+\cdots \nonumber\\
    &\quad\quad\quad\quad\quad
 \cdots+\hat\sigma^x_i\hat\sigma^x_{i-1}\cdots\hat\sigma^x_{i+2}\hat\sigma^z_{i+1}.
\end{align}
These operators commute with the local terms $\hat H_j$ except for $\hat H_i$ and $\hat H_{i+1}$.
The states $|\mathcal{O}^\pm_{i+(1/2)}\rangle$ are written as
\begin{align}
 |\mathcal{O}^+_{i+(1/2)}\rangle = \hat{\mathcal{O}}^+_{i+(1/2)}|0\cdots0\rangle,\nonumber\\
 |\mathcal{O}^-_{i+(1/2)}\rangle = \hat{\mathcal{O}}^-_{i+(1/2)}|0\cdots0\rangle.
\end{align}
The correlation function for $|\mathcal{O}^+_{i+(1/2)}\rangle$ and $|\mathcal{O}^-_{j+(1/2)}\rangle$ is computed as
\begin{align}
 \frac{|\langle\mathcal{O}^+_{i+(1/2)}|(\hat{\mathbbm{1}}-\hat G)|\mathcal{O}^-_{j+(1/2)}\rangle|}{\||\mathcal{O}^+_{i+(1/2)}\rangle\| \||\mathcal{O}^-_{j+(1/2)}\rangle\|}
 &
 = \frac{|\langle\mathcal{O}^+_{i+(1/2)}|\mathcal{O}^-_{j+(1/2)}\rangle|}{L-1} \nonumber\\ &
 = \frac{1}{L-1}.
\end{align}
Thus, by setting $|i-j| = \mathrm{const}\times L$, we can use Theorem~\ref{thm: general} to obtain the bound $z \ge 2$.

\subsubsection{Boundary terms on the XX model with a fine-tuned magnetic field}
Another example of hidden correlation is observed in the model of Eq.~\eqref{Hamiltonian of XX+MF} with the boundary terms:
\begin{align}
   \hat H = \sum_{i \in \partial\Lambda} \hat H_i + \sum_{\langle i,j \rangle}\hat H_{\langle i,j \rangle}.
\end{align}
Here, $\partial\Lambda$ consists of sites on the boundary of the lattice $\Lambda$ and the boundary terms are given by
\begin{align}
   \hat H_i = \frac1{2}\left(\hat{\mathbbm{1}}-\hat\sigma^z_i \right).
\end{align}
The index set $\mathcal{S}$ for this model is given by $\mathcal{S} = \{i \in \partial\Lambda\} \cup \{\langle i,j \rangle\}$.
One of the ground states in the original model
\begin{align}
 |\Psi_\Lambda\rangle \coloneqq \frac1{\sqrt{|\Lambda|}} \sum_{i \in \Lambda} \hat\sigma^-_i |\Psi_0\rangle
\end{align}
is penalized by the boundary terms.
(We changed the notation from $|\Psi_1\rangle$ to $|\Psi_\Lambda\rangle$ for later convenience.)
Thus, the unique ground state is given by
\begin{align}
 |\Psi_0\rangle = |0\cdots0\rangle.
\end{align}
This ground state exhibits no correlations between any
local operators.
However, there are power-law hidden correlation functions in an extended sense.
Specifically, there are correlation functions that converge to constants as the system size $L$ increases.

Let us consider two states $| \Psi_\Lambda \rangle$ and $| \Psi_R \rangle$, where $R$ is a subset of $\Lambda$ and $| \Psi_R \rangle$ is defined as
\begin{align}
 | \Psi_R \rangle = \frac1{\sqrt{|R|}}\sum_{i \in R}\hat\sigma_i^-| \Psi_0 \rangle.
\end{align}
The energy supports for these states $| \Psi_\Lambda \rangle$ and $| \Psi_R \rangle$ are localized around the boundary $\partial\Lambda$ and $\partial R$, respectively.
The correlation function between $| \Psi_\Lambda \rangle$ and $| \Psi_R \rangle$ is computed as
\begin{align}
 \frac{\langle\Psi_\Lambda|(\hat{\mathbbm{1}}-\hat G)|\Psi_R\rangle}{\|| \Psi_\Lambda \rangle\|\|| \Psi_R \rangle\|}
 = \langle\Psi_\Lambda|\Psi_R\rangle = \sqrt{\frac{|R|}{|\Lambda|}}.
 \label{corr func for XX+MF+Boundary}
\end{align}
Let $|R| = \mathrm{const}\times L^d$ and $D(| \Psi_\Lambda \rangle, | \Psi_R \rangle) = \mathrm{const}\times L$, where $d$ is the dimension of the lattice $\Lambda$.
For instance, if $\Lambda$ is the $d$-dimensional hypercubic lattice $\{1,\ldots,L\}^d$, a subset $R$ may be given by
\begin{align}
 R = \{\lfloor L/3\rfloor, \lfloor L/3\rfloor+1, \ldots, \lfloor2L/3\rfloor\}^d.
\end{align}
Then, the distance $D(| \Psi_\Lambda \rangle, | \Psi_R \rangle)$ is the distance between $\partial R$ and $\partial\Lambda$, which is proportional to $L$.
Thus, the correlation function in Eq.~\eqref{corr func for XX+MF+Boundary} is proportional to a constant, which allows us to apply Theorem~\ref{thm: general} and conclude that $z \ge 2$.

\section{Discusions}\label{sec: discussion}

This work establishes a rigorous lower bound $z\ge2$ on the dynamical exponent in frustration-free systems that exhibit power-law decaying correlation functions with increasing system size. This finding has significant implications for various classes of gapless FF systems.
These include Rokhsar-Kivelson Hamiltonians at critical points and systems with hidden algebraic correlations.
A key strength of our approach, based on the Gosset-Huang inequality, is its generality; it uniformly treats these diverse quantum critical points without specific constraints on dimensionality, lattice structure, or boundary conditions.

Furthermore, by leveraging the established correspondence between RK Hamiltonians and Markov processes, we demonstrate that $z\ge2$ also applies to Markov processes with local state updates that satisfy the detailed balance condition. This result theoretically substantiates an empirical observation prevalent in non-equilibrium statistical mechanics and computational physics. The application of the Gosset-Huang inequality, a tool originating from quantum information theory, highlights the value of a quantum information perspective for advancing our understanding of dynamical critical phenomena.

Despite the generality of the presented proof, several avenues for future investigation remain. A crucial open question is whether gapless FF systems universally exhibit power-law or slower decaying correlation functions, or if counterexamples exist. For instance, in the ordered phase of the kinetic Ising model for dimensions $d>2$, a power-law correlation function has not yet been identified, even though the system is expected to exhibit phase ordering kinetics with $z=2$~\cite{brayTheoryPhaseOrdering1994}. Conceivably, an extended definition of correlation functions, such as the one explored in Section~\ref{sec: hidden}, could address these cases.

Another important direction is to investigate the robustness of the atypical dynamical exponent in gapless FF systems against perturbations that break the FF condition.
This question pertains to the broader issue of how readily the distinct behaviors of gapless FF systems manifest in physical systems.
A universal understanding of this problem may require a comprehensive field-theoretic framework. However, elucidating the fundamental origins of both the FF condition and the Gosset-Huang inequality within a continuum field theory remains a major open problem.

The rich interplay between frustration-freeness, correlations, and dynamics suggests that gapless FF systems will remain a fertile ground for discovering novel phenomena and deepening our understanding of quantum and classical critical dynamics. Further theoretical and numerical explorations are expected to unravel the remaining intricacies of these fascinating systems.

\section{Acknowledgements}

We thank Hidemaro Suwa, Carolyn Zhang, Ruben Verresen, Jun Takahashi and Hal Tasaki for the useful discussions.
The work of H.W. is supported by JSPS KAKENHI Grant No.~JP24K00541.
This research is funded in part by the
Gordon and Betty Moore Foundation's EPiQS Initiative,
Grant No. GBMF8683 to T.S.

\appendix

\section{Proof of the Gosset-Huang inequality}
\label{sec: proof of GH inequality}
In this section, we prove the Gosset-Huang inequality, following the arguments in Refs.~\cite{gossetCorrelationLengthGap2016, anshuSimpleProofDetectability2016}.

\subsection{Detectability lemma}
First, we show the following lemma.
\begin{lmm}[Detectability lemma~\cite{aharonovDetectabilityLemmaIts2011,anshuSimpleProofDetectability2016}]
 \label{Detectability lemma}
 Let $\hat H = \sum_{i \in \mathcal{S}} \hat H_i$ be a FF Hamiltonian and let local terms $\{\hat H_i\}_{i \in \mathcal{S}}$ be orthogonal projectors.
 We define $\hat P_i \coloneqq \hat{\mathbbm{1}} - \hat H_i$ and $\hat P \coloneqq \prod_{i\in\mathcal{S}} \hat P_i$, where the product is taken in any order.
 Then, the following inequality holds:
 \begin{align}
 \|\hat P-\hat G\| \le \sqrt{\frac{g^2}{\epsilon+g^2}}.
 \label{detectability lemma (Appendix)}
 \end{align}
 Here, $\hat G$ is the projector onto the ground space of $\hat H$,  $\epsilon$ is the spectral gap of $\hat H$, and $g$ is the maximum degree of the interaction graph.
\end{lmm}

\begin{proof}
We follow the argument in Ref.~\cite{anshuSimpleProofDetectability2016}.
We label the indices $i\in \mathcal{S}$ by $1, 2, \dots, N$, where $N = |\mathcal{S}|$.
This labeling is entirely arbitrary and carries no information about the spatial arrangement of the sites.
Let $\hat P_i = \hat{\mathbbm{1}}-\hat H_i$ and $\hat P = \hat P_1\hat P_2\cdots\hat P_{N}$.
These operators satisfy the following relations since $\hat H_i\hat G = 0$:
\begin{align}
\hat P_i\hat G = \hat G\hat P_i = \hat G,\quad
\hat P\hat G = \hat G\hat P = \hat G.
\label{PG=GP=G}
\end{align}
Let us consider the quantity $ \|\hat H_i\hat P_{j}\cdots \hat P_{N}|\psi\rangle\| $ for any state $|\psi\rangle$ and perform the following procedure.
If $\hat H_i$ commute with $\hat H_{j}$ (and, hence, with $\hat P_j$), we use 
\begin{align}
\|\hat H_i\hat P_{j}\cdots \hat P_{N}|\psi\rangle\| &
= \|\hat P_{j}\hat H_i\hat P_{j+1}\cdots \hat P_{N}|\psi\rangle\|\nonumber\\
&
\le \|\hat H_i\hat P_{j+1}\cdots \hat P_{N}|\psi\rangle\|.
\end{align}
Otherwise, we recall $\hat P_{j} = \hat{\mathbbm{1}}-\hat H_{j}$ to see
\begin{align}&
\| \hat H_i\hat P_{j}\hat P_{j+1}\cdots \hat P_{N}|\psi\rangle\|\nonumber\\
&
\le \| \hat H_i\hat P_{j+1}\cdots \hat P_{N}|\psi\rangle\|
+ \|\hat H_i\hat H_{j}\hat P_{j+1}\cdots \hat P_{N}|\psi\rangle\| \nonumber\\
&
\le \|\hat H_i\hat P_{j+1}\cdots \hat P_{N}|\psi\rangle\|
+ \|\hat H_{j}\hat P_{j+1}\cdots \hat P_{N}|\psi\rangle\|.
\end{align}
Since we assumed that $\hat H_i$ is a projector, we have
\begin{align}
\hat H_i\hat P_i = \hat H_i(\hat{\mathbbm{1}}-\hat H_i) = 0.
\end{align}
This implies that the above procedure stops when $j = i$ before $\hat H_i$ reaches the rightmost position.
If we start from $\| \hat H_i\hat P|\psi\rangle\|$ and repeat the above procedure, we obtain
\begin{align}
\|\hat H_i\hat P|\psi\rangle\|
\le \sum_{l \in B_i} \|\hat H_{l}\hat P_{l+1}\cdots\hat P_{N}|\psi\rangle\|,
\end{align}
where $B_i \coloneqq \{j\mid [\hat H_i, \hat H_j] \ne 0 \}$ is the set of indices adjacent to $i$ in the interaction graph.
When $l=N$, the product $\hat{P}_{l+1}\cdots \hat{P}_N$ represents the identity.
The number of terms in the sum is given by the degree $g_i = |B_i|$ of the vertex $i$.
Since the square of the average can be bounded above by the average of the square, we find
\begin{align}
\| \hat H_i\hat P|\psi\rangle\|^2
&
\le g_i^2\left(\frac1{g_i} \sum_{l \in B_i} \|\hat H_{l} \hat P_{l+1} \cdots \hat P_{N} |\psi\rangle\|\right)^2 \nonumber\\
&
\le g_i^2\frac{1}{g_i} \sum_{l \in B_i} \|\hat H_{l} \hat P_{l+1} \cdots \hat P_{N}|\psi\rangle \|^2 \nonumber\\
&
\le g\sum_{l \in B_i} \|\hat H_{l} \hat P_{l+1} \cdots \hat P_{N}|\psi\rangle \|^2.
\end{align}
Summing this inequality from $i=1$ to $N$ yields
\begin{align}&
\langle\psi|\hat P^\dagger\hat H\hat P|\psi\rangle = \sum_{i = 1}^N \| \hat H_i\hat P|\psi\rangle\|^2\nonumber\\
&
\le g \sum_{i = 1}^N \sum_{l \in B_i} \|\hat H_{l} \hat P_{l+1} \cdots \hat P_{N}|\psi\rangle \|^2 \nonumber\\
&
= g \sum_{l = 1}^N \sum_{i \in B_l} \|\hat H_{l} \hat P_{l+1} \cdots \hat P_{N}|\psi\rangle \|^2 \nonumber\\
&
\le g^2 \sum_{l = 1}^{N} \|\hat H_{l}\hat P_{l+1} \cdots \hat P_{N}|\psi\rangle \|^2 \nonumber\\
&
= g^2\sum_{l = 1}^{N} \langle\psi|\hat P_{N}\cdots \hat P_{l+1}(\hat{\mathbbm{1}}-\hat P_{l})\hat P_{l+1} \cdots \hat P_{N}|\psi\rangle\nonumber\\
&
= g^2\sum_{l = 1}^{N} (\| \hat{P}_{l+1}\cdots\hat{P}_N|\psi\rangle \|^2 - \| \hat{P}_l\cdots\hat{P}_N|\psi\rangle \|^2 )\nonumber\\
&
= g^2(\||\psi\rangle\|^2 -\|\hat P|\psi\rangle\|^2).
\end{align}
Let $|\psi^\bot\rangle = \hat{G}^\perp \coloneqq (\hat{\mathbbm{1}}-\hat G)|\psi\rangle$.
Recalling that $\hat P\hat G = \hat G$ in Eq.~\eqref{PG=GP=G},
we find that the state $\hat P|\psi^\bot\rangle$ is orthogonal to the ground states since $\hat G\hat P|\psi^\bot\rangle = \hat G|\psi^\bot\rangle = 0$.
Therefore,
\begin{align}
\epsilon \| \hat P|\psi^\bot\rangle\|^2
&
\le \langle\psi^\bot|\hat P^\dagger \hat H\hat P|\psi^\bot\rangle \nonumber\\
&
\le g^2(\||\psi^\bot\rangle\|^2-\| \hat P|\psi^\bot\rangle\|^2),
\end{align}
where $\epsilon$ is the spectral gap of $\hat H$.
Thus, we obtain
\begin{align}
\|(\hat P-\hat G)|\psi\rangle\| &= \|\hat P(\hat{\mathbbm{1}}-\hat G)|\psi\rangle\| = \|\hat{P}|\psi^\bot\rangle\| \nonumber\\
&
\le \sqrt{\frac{g^2}{g^2+\epsilon}} \||\psi^\bot\rangle\| \nonumber\\
&
\le \sqrt{\frac{g^2}{g^2+\epsilon}} \||\psi\rangle\|.
\end{align}
Here, we used Eq.~\eqref{PG=GP=G} again.
This implies Eq.~\eqref{detectability lemma (Appendix)}.
\end{proof}

\subsection{Gosset-Huang inequality}

Let us prove the following Gosset-Huang inequality in an extended form.
\begin{thm}
Let $\hat H = \sum_{i \in \mathcal{S}} \hat H_i$ be a FF Hamiltonian
~\footnote{
If the Hamiltonian is not FF, the same bound as in Eq.~\eqref{GH inequality} holds by replacing $\gap(\hat H)$ by the ground-state energy.
However, this inequality is trivial in most cases due to the absence of distant states.
} with positive semidefinite local terms $\{\hat H_i\}_{i \in \mathcal{S}}$.
For any states $|\mathcal{O}\rangle$ and $|\mathcal{O}'\rangle$, the following inequality holds.
\begin{align}&
 \frac{|\langle\mathcal{O}|(\hat{\mathbbm{1}}-\hat G)|\mathcal{O}'\rangle|}{\||\mathcal{O}\rangle\| \||\mathcal{O}'\rangle\|} \nonumber\\
    &
 \le 2 \exp\left(-\left( \frac{D(|\mathcal{O}\rangle,|\mathcal{O}'\rangle)-1}{c-1} -2 \right)\sqrt{\frac{\epsilon}{g^2+\epsilon}}\right).
   \label{GH inequality}
\end{align}
Here, $\hat G$ is the orthogonal projector onto $\ker \hat H$, 
$D(|\mathcal{O}\rangle,|\mathcal{O}'\rangle)$ is the distance between $|\mathcal{O}\rangle$ and $|\mathcal{O}'\rangle$ defined in Sec.~\ref{sec: GH inequality}, 
and $\epsilon$ is the dimensionless spectral gap given by
\begin{align}
 \epsilon \coloneqq \frac{\gap(\hat H)}{\max_i \|\hat H_i\|}.
\end{align}
Also, $g$ is the maximum degree of the interaction graph, and $c$ is the chromatic number of the interaction graph.
The definitions of these integers are given in Sec.~\ref{sec: GH inequality}.
\end{thm}

The original form of the Gosset-Huang inequality in Eq.~\eqref{original GH inequality} can be derived by substituting $\langle \mathcal{O} | = \langle \Psi |\hat{\mathcal{O}}$ and $| \mathcal{O}' \rangle = \hat{\mathcal{O}}'| \Psi \rangle$ into Eq.~\eqref{GH inequality}, where $| \Psi \rangle$ is a ground state of the Hamiltonian:
\begin{align}&
 \frac{|\langle\mathcal{O}|(\hat{\mathbbm{1}}-\hat G)|\mathcal{O}'\rangle|}{\||\mathcal{O}\rangle\| \||\mathcal{O}'\rangle\|}
 = \frac{|\langle\Psi|\hat{\mathcal{O}}(\hat{\mathbbm{1}}-\hat G)\hat{\mathcal{O}}'|\Psi\rangle|}{\|\hat{\mathcal{O}}^\dagger|\Psi\rangle\| \|\hat{\mathcal{O}}'|\Psi\rangle\|} \nonumber\\
   &
 \le 2 \exp\left[-\left( \frac{D(|\mathcal{O}\rangle, |\mathcal{O}'\rangle)-1}{c-1} -2 \right)\sqrt{\frac{\epsilon}{g^2+\epsilon}}\right] \nonumber\\
   &
 \le 2 \exp\left[-\left( \frac{D(\hat{\mathcal{O}}, \hat{\mathcal{O}}')-1}{c-1} -2 \right)\sqrt{\frac{\epsilon}{g^2+\epsilon}}\right].
\end{align}
Here, we have used the inequality
\begin{align}
 D(\hat{\mathcal{O}}, \hat{\mathcal{O}}')
   &
 = \min\{d(i,j)\mid [\hat{\mathcal{O}}, \hat H_i] \ne 0,~[\hat{\mathcal{O}}',\hat H_j] \ne 0\} \nonumber\\
   &
 \le \min\{d(i,j)\mid \hat H_i| \mathcal{O} \rangle \ne 0,~ \hat H_j| \mathcal{O}' \rangle \ne 0\} \nonumber\\
   &
 = D(|\mathcal{O}\rangle, |\mathcal{O}'\rangle).
\end{align}

\begin{proof}
 We begin by assuming that each local term $\hat H_i$ is an orthogonal projector.
This assumption is relaxed in later discussions.

 We divide the vertices of the interaction graph into $c$ colors so that no two adjacent vertices have the same color, which means that the same colored $\hat H_i$ commute with each other.
 We denote the local term corresponding $i$th vertex with color $j$ as $\hat H_i^{(j)}$ and consider
    \begin{align}
        \hat P \coloneqq \prod_{i_c} \hat P_{i_c}^{(c)} \prod_{i_{c-1}} \hat P_{i_{c-1}}^{(c-1)}\cdots\prod_{i_2} \hat P_{i_2}^{(2)}\prod_{i_1} \hat P_{i_1}^{(1)},
    \end{align}
 where $\hat P_i^{(j)} \coloneqq \hat{\mathbbm{1}} - \hat H_i^{(j)}$.
 A key observation is that in the expression $\langle\mathcal{O}|(\hat P^\dagger \hat P)^n$, only the contributions from the operators $\hat P_i$ within the ``light cone'' shown in Fig.~\ref{fig: light cone} remain:
\begin{figure}[t]
\centering
\includegraphics[width=0.7\hsize]{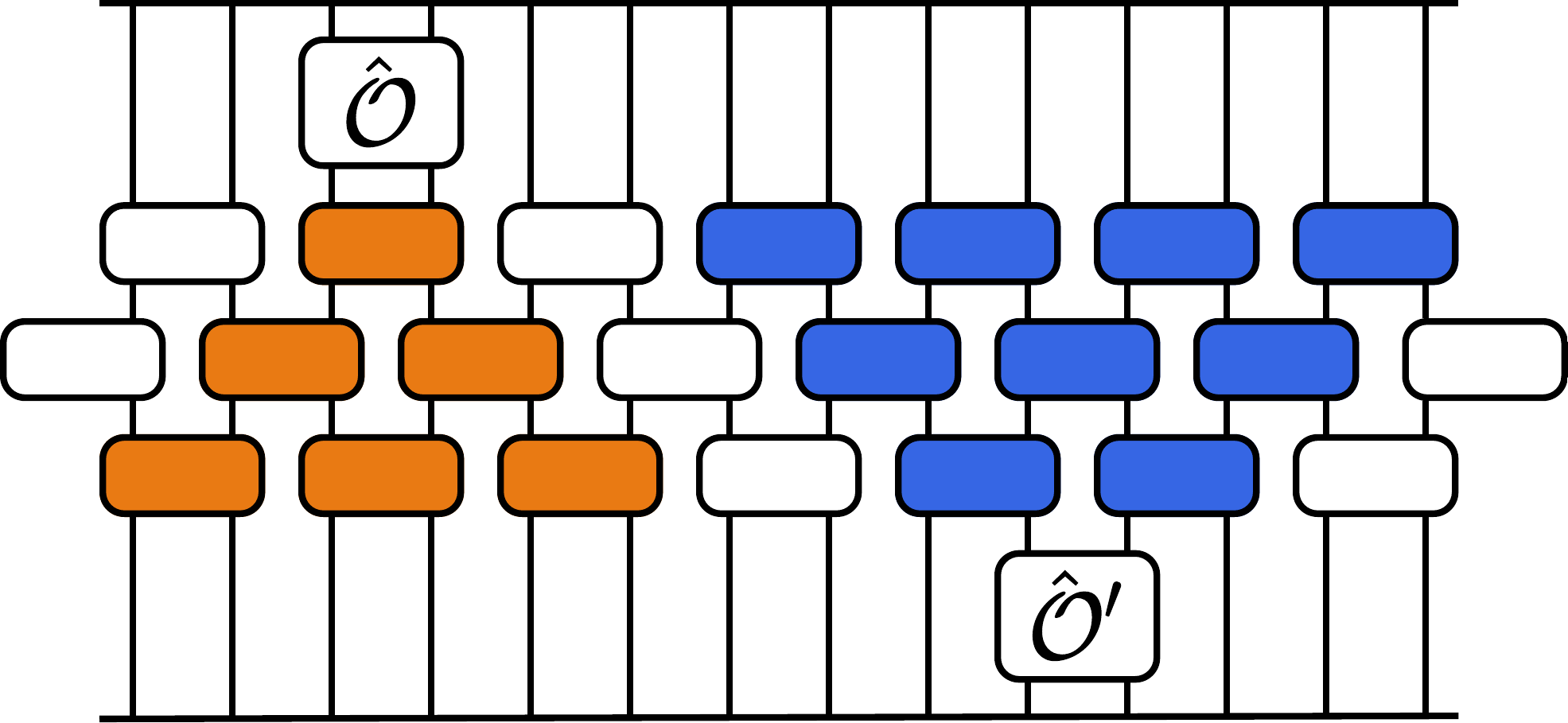}
\caption{Two light cones from $\langle\Psi|\hat{\mathcal{O}}$ (orange) and $\hat{\mathcal{O}'}|\Psi\rangle$ (blue) in a one-dimensional lattice.}
\label{fig: light cone}
\end{figure}
\begin{align}&
\langle\mathcal{O}|(\hat P^\dagger \hat P)^n \nonumber\\
&
= \langle\mathcal{O}| \prod_{i_1} \hat P_{i_1}^{(1)} \prod_{i_2}\hat P_{i_2}^{(2)}\prod_{i_3}\hat P_{i_3}^{(3)}\cdots \nonumber\\
&
= \langle\mathcal{O}|
\!\!\!\prod_{i_1:\,D(i_1,|\mathcal{O}\rangle)=0}\!\!\!\hat P_{i_1}^{(1)}
\!\!\!\prod_{i_2:\,D(i_2,|\mathcal{O}\rangle)\le1}\!\!\!\hat P_{i_2}^{(2)}
\!\!\!\prod_{i_3:\,D(i_3,|\mathcal{O}\rangle)\le2}\!\!\!\hat P_{i_3}^{(3)} \cdots,
\end{align}
 where 
\begin{align}
 D(i, |\mathcal{O}\rangle) &\coloneqq \min\{d(i, j) \mid j \in \mathcal{S}(|\mathcal{O}\rangle)\} \nonumber\\
   &
 = \min\{d(i,j) \mid \hat H_j| \mathcal{O} \rangle \ne 0\}.
\end{align}
A similar equation holds for $(\hat P^\dagger\hat P)^n|\mathcal{O}'\rangle$.
Therefore, we have
\begin{align}
 \langle\mathcal{O}|\mathcal{O}'\rangle = \langle\mathcal{O}|(\hat P^\dagger\hat P)^n|\mathcal{O}'\rangle,
   \label{eq: insertion of DL op}
   \end{align}
as long as the two light cones from $\langle\mathcal{O}|$ and $|\mathcal{O}'\rangle$ do not overlap.
The integer $n$ must satisfy
\begin{align}
   2n(c-1)+1 \le D(|\mathcal{O}\rangle, |\mathcal{O}'\rangle),
\end{align}
because the product $(\hat P^\dagger\hat P)^n$ has $2n(c-1)+1$ layers.
The number of layers is derived from the following idempotent properties:
\begin{align}
 \left( \prod_{i_1} \hat P_{i_1}^{(1)} \right)^2 = \prod_{i_1} \hat P_{i_1}^{(1)},\quad
 \left( \prod_{i_c} \hat P_{i_c}^{(c)} \right)^2 = \prod_{i_c} \hat P_{i_c}^{(c)}.
\end{align}
Consequently, the maximum $n$ in Eq.~\eqref{eq: insertion of DL op} is given by
\begin{align}
 n \le m \coloneqq \left\lfloor \frac{D(|\mathcal{O}\rangle, |\mathcal{O}'\rangle)-1}{2(c-1)} \right\rfloor.
\end{align}
 Therefore, for polynomials $Q_m(x)$ such that $\deg Q_m(x) \le m$ and $Q_m(1) = 1$, we have
\begin{align}
 \langle\mathcal{O}|\mathcal{O}'\rangle
 = \langle\mathcal{O}|Q_m(\hat P^\dagger \hat P)|\mathcal{O}'\rangle.
\end{align}
Recalling that $\hat P\hat G = \hat G\hat P = \hat G$ in Eq.~\eqref{PG=GP=G},
we find for $n > 0$ that
\begin{align}
 (\hat P^\dagger\hat P-\hat G)^n = (\hat P^\dagger \hat P)^n-\hat G.
\end{align}
We also see
    \begin{align}
 (\hat P^\dagger\hat P-\hat G)^n\hat G = 0.
    \end{align}
Combining these equations, we obtain
   \begin{align}
(\hat P^\dagger \hat P)^n -\hat G = (\hat P^\dagger\hat P - \hat G)^n\hat G^\bot,
   \end{align}
where $\hat G^\bot \coloneqq \hat{\mathbbm{1}}-\hat G$.
This equation also holds for $n=0$.
Therefore, we obtain
\begin{align}
 |\langle\mathcal{O}|(\hat{\mathbbm{1}}-\hat G)|\mathcal{O}'\rangle|
      &
 = |\langle\mathcal{O}|(Q_m(\hat P^\dagger\hat P)-\hat G)|\mathcal{O}'\rangle| \nonumber\\
      &
 = |\langle\mathcal{O}|Q_m(\hat P^\dagger\hat P-\hat G)\hat G^\bot|\mathcal{O}'\rangle| \nonumber\\
      &
 \le \||\mathcal{O}\rangle \| \||\mathcal{O}'\rangle \| \| Q_m(\hat P^\dagger\hat P-\hat G)\|.
   \label{bound of correlation function}
\end{align}
From the detectability lemma~\ref{Detectability lemma}, the operator norm of $\hat P^\dagger\hat P-\hat G$ is upper bounded as
\begin{align}
 \|\hat P^\dagger\hat P-\hat G\|
         &
 = \|(\hat P - \hat G)^\dagger(\hat P-\hat G)\| 
 = \|\hat P-\hat G\|^2 \nonumber\\
         &
 \le \frac{g^2}{g^2+\epsilon} \eqcolon 1 - \delta,
\end{align}
where $\delta \coloneqq \epsilon/(g^2+\epsilon)$.
Thus, the norm of a polynomial $Q_m(\hat P^\dagger\hat P-\hat G)$ is bounded as
\begin{align}
 \| Q_m(\hat P^\dagger\hat P-\hat G)\| \le \max_{0 \le x \le 1-\delta} |Q_m(x)|.
   \label{mini max problem}
\end{align}
We minimize the right-hand side of Eq.~\eqref{mini max problem} under the constraint $\deg Q_m \le m$ and $Q_m(1) = 1$.
The optimal polynomial is given by
\begin{align}
 Q_m(x) = \frac{T_m(\frac{2x}{1-\delta}-1)}{T_m(\frac2{1-\delta}-1)},
\end{align}
where $T_m(x)$ is the degree $m$ Chebyshev polynomial of the first kind defined by
\begin{align}
   \cos(mx) = T_m(\cos x).
\end{align}
For $x > 1$, the polynomial $T_m(x)$ satisfies
\begin{align}
 T_m(x)
   &
 = \cosh(m\arccosh x) \nonumber\\
   &
 >\frac1{2}\exp(m\arccosh x) \nonumber\\
   &
 \ge\frac1{2}\exp\left( 2m\tanh\left(\frac1{2}\arccosh x\right)\right) \nonumber\\
   &
 = \frac1{2}\exp\left( 2m\sqrt{\frac{x-1}{x+1}} \right).
\end{align}
 Also $|T_m(x)| \le 1$ for $|x| \le 1$.
Therefore, for $0 \le x \le 1-\delta$, $Q_m(x)$ satisfies
    \begin{align}
 |Q_m(x)| &\le  \frac1{|T_m(\frac2{1-\delta}-1)|} \nonumber\\
 &
 \le \left\lbrack \frac1{2}\exp\left( 2m\sqrt{\frac{\frac{2}{1-\delta}-2}{\frac{2}{1-\delta}}} \right)\right\rbrack^{-1} \nonumber\\
    &
 = 2\exp\left(-2m\sqrt{\frac{\epsilon}{g^2+\epsilon}}\right).
    \label{bound of Q_m(x)}
    \end{align}
 From Eq.~\eqref{bound of correlation function}, Eq.~\eqref{mini max problem} and Eq.~\eqref{bound of Q_m(x)}, we obtain
\begin{align}&
 \frac{\langle\mathcal{O}|(\hat{\mathbbm{1}}-\hat G)|\mathcal{O}'\rangle}{\||\mathcal{O}\rangle\| \||\mathcal{O}'\rangle\|} 
 \le 2 \exp\left(-2m\sqrt{\frac{\epsilon}{g^2+\epsilon}}\right) \nonumber\\
   &
 = 2\exp\left( -2\left\lfloor\frac{D(|\mathcal{O}\rangle,|\mathcal{O}'\rangle)-1}{2(c-1)}\right\rfloor\sqrt{\frac{\epsilon}{g^2+\epsilon}} \right) \nonumber\\
   &
 \le 2\exp\left( -\left( \frac{D(|\mathcal{O}\rangle,|\mathcal{O}'\rangle)-1}{c-1}-2 \right)\sqrt{\frac{\epsilon}{g^2+\epsilon}} \right).
   \label{GH for projector}
\end{align}

We note the exceptional cases where the distance $D(|O\rangle, |O'\rangle)$ is infinite. These occur when the interaction graph contains no edges, or when there exists no path of noncommuting local terms connecting the two states $|\mathcal{O}\rangle$ and $|\mathcal{O}'\rangle$. In such cases, the correlation function in Eq.~\eqref{GH inequality} vanishes, and the Gosset-Huang inequality holds with our formal convention $D(|\mathcal{O}\rangle, |\mathcal{O}'\rangle) =\infty$.

Finally, we consider the case where local terms $\hat H_i$ are not orthogonal projectors.
The local terms $\hat H_i$ are assumed to be positive semidefinite, which can always be achieved by adding a suitable constant multiple of the identity operator to each term.
In this case, we can deform the local terms $\hat H_i$ into orthogonal projectors $\hat{\tilde H}_i$ while preserving $\ker \hat H_i$.
The deformed Hamiltonian $\hat{\tilde H} = \sum_i \hat{\tilde H}_i$ has the same ground space as $\hat H$.
These two Hamiltonians also satisfy
\begin{align}
 \frac{\hat H}{\max_i \|\hat H_i\|}
 \le \sum_{i \in \mathcal{S}} \frac{\hat H_i}{\|H_i\|}
 \le \sum_{i \in \mathcal{S}} \hat{\tilde H}_i = \hat{\tilde H}.
\end{align}
Therefore, the dimensionless spectral gap $\epsilon$ of the Hamiltonian $\hat H$ and the spectral gap $\tilde{\epsilon}$ of the Hamiltonian $\hat{\tilde H}$ satisfy
\begin{align}
 \epsilon \coloneqq \frac{\gap(\hat H)}{\max_i \|\hat H_i\|} \le \gap(\hat{\tilde H}) \eqcolon \tilde{\epsilon}.
   \label{inequality for spectral gap}
\end{align}
Furthermore, this deformation may modify the interaction graph.
If two local terms $\hat H_i$ and $\hat H_j$ commute with each other, these two operators can be simultaneously diagonalized:
\begin{align}
   \hat H_i = \sum_a E_{ia} | a \rangle\langle a |,\quad
   \hat H_j = \sum_a E_{ja} | a \rangle\langle a |,
\end{align}
where $\{| a \rangle\}_a$ is an orthonormal basis and $\{E_{ia}\}_a$ and $\{E_{ja}\}_a$ are the eigenvalues of $\hat H_i$ and $\hat H_j$, respectively.
Since our deformation changes only eigenvalues $\{E_{ia}\}_a$ and $\{E_{ja}\}_a$ into values of $0$ or $1$ and does not change the basis $\{| a \rangle\}_a$, the operators $\hat{\tilde H}_i$ and $\hat{\tilde H}_j$ also commute with each other.
In other words, if $i$ and $j$ are not connected in the interaction graph of $\{\hat H_i\}$, they are also not connected in that of $\{\hat{\tilde H}_i\}$.
However, the converse is not generally true.
Therefore, the maximum degree $g$ and chromatic number $c$ for $\{\hat H_i\}$, and their counterparts $\tilde g$ and $\tilde c$ for $\{\hat{\tilde H}_i\}$, satisfy
\begin{align}
 g \ge \tilde g,\quad\text{and}\quad c \ge \tilde c.
   \label{inequality for g and c}
\end{align}
Using Eqs.~\eqref{inequality for spectral gap} and~\eqref{inequality for g and c}, and the Gosset-Huang inequality for projector local terms in Eq.~\eqref{GH for projector}, we obtain
\begin{align}&
 \frac{\langle\mathcal{O}|(\hat{\mathbbm{1}}-\hat G)|\mathcal{O}'\rangle}{\||\mathcal{O}\rangle\| \||\mathcal{O}'\rangle\|} \nonumber\\
   &
 \le 2\exp\left[ -\left( \frac{D(|\mathcal{O}\rangle,|\mathcal{O}'\rangle)-1}{\tilde c-1}-2 \right)\sqrt{\frac{\tilde{\epsilon}}{{\tilde g}^2+\tilde{\epsilon}}} \right] \nonumber\\
   &
 \le 2\exp\left[ -\left( \frac{D(|\mathcal{O}\rangle,|\mathcal{O}'\rangle)-1}{\tilde c-1} -2\right)\sqrt{\frac{\epsilon}{\tilde{g}^2+\epsilon}} \right] \nonumber\\
   &
 \le 2\exp\left[ -\left( \frac{D(|\mathcal{O}\rangle,|\mathcal{O}'\rangle)-1}{c-1} -2\right)\sqrt{\frac{\epsilon}{g^2+\epsilon}} \right].
\end{align}
 Here, we used that $\sqrt{x/(\tilde{g}^2+x)}$ is a monotonically increasing function of $x$ for $x \ge 0$.
\end{proof}

\section{Construction of RK Hamiltonians from Markov processes} \label{sec: MCMC}
\subsection{The Gibbs sampling algorithm}
As an example of the construction of RK Hamiltonians, we introduce the Gibbs sampling (heat-bath) algorithm~\cite{gemanStochasticRelaxationGibbs1984}, one of the most standard MCMC methods.
We divide the system into the $R_i$: a finite-size region around a position $i$, $B_i$: the boundary of $R_i$ interacting with $R_i$, and $E_i$: the external system not interacting with $R_i$.
Then, we can decompose the Boltzmann weight as
\begin{align}
 w(C) = w^{R_i}(C_{R_i},C_{B_i}) w^{E_i}(C_{B_i},C_{E_i}), \label{locality of Boltzmann weight}
\end{align}
where $C_{R_i},C_{B_i},$ and $C_{E_i}$ are the configurations in $R_i,B_i,$ and $E_i$, respectively.
We consider the Gibbs sampling algorithm that updates the variables in $R_i$ at each step.
The local transition-rate matrices are given by
\begin{align}
    \hat W_i = \sum_{C_{B_i}}\hat W_i(C_{B_i}) \otimes |C_{B_i}\rangle\langle C_{B_i}| \otimes \hat{\mathbbm{1}}_{E_i} ,
\end{align}
where
\begin{align}
 \langle C_{R_i}|\hat W_i(C_{B_i})|C'_{R_i}\rangle 
 &
 = \frac{w^{R_i}(C_{R_i},C_{B_i})}
 {\sum_{C''_{R_i}}w^{R_i}(C''_{R_i},C_{B_i})}
 - \delta_{C_{R_i} C'_{R_i}}.
\end{align}
This transition-rate matrix conserves probability and satisfies the detailed balance condition.
The corresponding local terms of the Hamiltonian are orthogonal projectors written as
\begin{align}
    \hat{H}_i = \left( \hat{\mathbbm{1}}_{R_i} \otimes \hat{\mathbbm{1}}_{B_i} - \sum_{C_{B_i}} |\Psi(C_{B_i})\rangle\langle\Psi(C_{B_i})| \right) \otimes\hat{\mathbbm{1}}_{E_i},
\end{align}
where
\begin{align}&
 |\Psi(C_{B_i})\rangle \coloneqq \sum_{C_{R_i}}\sqrt{\frac{ w^{R_i}(C_{R_i}, C_{B_i})}{Z(C_{B_i})}}|C_{R_i}\rangle \otimes |C_{B_i}\rangle, \\
    &
 Z(C_{B_i}) \coloneqq \sum_{C_{R_i}} w^{R_i}(C_{R_i}, C_{B_i}).
\end{align}

\subsection{The Metropolis-Hastings algorithm}
Another widely used MCMC method is the Metropolis-Hastings algorithm~\cite{metropolisEquationStateCalculations1953,hastingsMonteCarloSampling1970}.
The off-diagonal elements of the transition-rate matrix are given by
\begin{align}
 W_{CC'} =  A(C,C')Q(C|C'),
\end{align}
for $C \ne C'$.
The diagonal elements are determined so that probability is conserved:
\begin{align}
 W_{CC} = -\sum_{C'\ne C}A(C',C)Q(C'|C).
\end{align}
A proposal distribution $Q(C'|C)$ takes non-negative values, and $A(C',C)$ is defined as
\begin{align}
 A(C,C') \coloneqq \min\left(1,\frac{w(C)Q(C'|C)}{w(C')Q(C|C')}\right).
\end{align}
The transition-rate matrix satisfies the detailed balance condition, since
\begin{align}
 W_{CC'} w(C')
    &
 = \min\left( w(C')Q(C|C'),w(C)Q(C'|C)\right) \nonumber\\
    &
 = W_{C'C} w(C).
    \label{DBC for Metropolis-Hastings}
\end{align}
The corresponding RK Hamiltonian is given by
\begin{align}
 \langle C|\hat H|C'\rangle
    &
 =  -\sqrt{\frac{w(C')}{w(C)}}A(C,C')Q(C|C')
\end{align}
for $C \ne C'$ and $\langle C|\hat H|C\rangle = - W_{CC}$.
This Hamiltonian is local if the given Boltzmann weight $w(C)$ and the proposal distribution $Q(C'|C)$ are local.

\section{Improvement of the bound for the critical kinetic Ising model}
In the proof of Theorem~\ref{thm: main}, we showed that, when the ground state exhibits power-law correlations, the spectral gap satisfies the bound $\epsilon \le O[(\log L)^2/L^2]$.
When the relevant correlation function instead approaches a nonzero constant, the bound sharpens to $\epsilon \le O(1/L^2)$.
Several FF models indeed exhibit such constant correlations for suitably chosen extended operators.
We now verify, using heuristic scaling arguments, that the spectral gap satisfies $\epsilon \le O(1/L^2)$ for the $2+1$-dimensional critical kinetic Ising model with arbitrary boundary conditions.

Let $C$ be a closed loop on the lattice, and define the corresponding loop operator
\begin{align}
\hat\sigma^z_{C} \coloneqq \sum_{\bm{x} \in C} \hat\sigma^z_{\bm{x}},
\end{align}
where the index $\bm{x}$ denotes the lattice sites.
We now consider two concentric circular loops centred at the lattice midpoint: $C$ with radius $L/8$ and $C'$ with radius $L/4$.
The correlation function between these two loop operators is
\begin{align}
 \frac{\langle \Psi|\hat\sigma^z_{C}(\hat{\mathbbm{1}}-\hat{G})\hat\sigma^z_{C'}| \Psi \rangle }{\|\hat\sigma^z_{C}|\Psi\rangle\| \|\hat\sigma^z_{C'}|\Psi\rangle\|}.
\end{align}
To analyze its scaling, we introduce the rescaled coordinate $\bm{y} \coloneqq \bm{x}/L$ inside the unit square and a continuum spin operator $\hat{\sigma}^z(\bm{y}) \coloneqq L^{\Delta_\sigma}\hat{\sigma}^z_{L\bm{y}}$, where $\Delta_\sigma = 1/8$.
With these definitions, correlation functions for the continuum spin $\hat\sigma^z(\bm{y})$ converge to finite values in the $L\to\infty$ limit.
Each factor in the correlation function can then be written as
\begin{align}&
 \langle \Psi|\hat\sigma^z_{C}(\hat{\mathbbm{1}}-\hat{G})\hat\sigma^z_{C'}| \Psi \rangle \nonumber\\
   &
 = L^{7/4}\oint_C \mathrm{d}{y} \oint_{C'} \mathrm{d}{y'} ( \langle\sigma(\bm{y})\sigma(\bm{y}')\rangle - \langle\sigma(\bm{y})\rangle\langle\sigma(\bm{y}')\rangle ),\\
   &
 \langle \Psi|\hat\sigma^z_{C}\hat\sigma^z_{C}| \Psi \rangle = L^{7/4}\oint_{C} \mathrm{d}{y} \oint_{C} \mathrm{d}{y'} \langle\sigma(\bm{y})\sigma(\bm{y}')\rangle, \\
   &
 \langle \Psi|\hat\sigma^z_{C'}\hat\sigma^z_{C'}| \Psi \rangle = L^{7/4}\oint_{C'} \mathrm{d}{y} \oint_{C'} \mathrm{d}{y'} \langle\sigma(\bm{y})\sigma(\bm{y}')\rangle.
\end{align}
Because each integral contains at most an integrable singularity of the form $1/|\bm{y}-\bm{y}'|^{1/4}$ at $\bm{y}=\bm{y}'$, none of them diverge.
Consequently, all three factors scale as $L^{7/4}$.
Hence, for sufficiently large $L$, we obtain
\begin{align}
\frac{\langle \Psi|\hat\sigma^z_{C}(\hat{\mathbbm{1}}-\hat{G})\hat\sigma^z_{C'}| \Psi \rangle }{\|\hat\sigma^z_{C}|\Psi\rangle\| \|\hat\sigma^z_{C'}|\Psi\rangle\|} = O(1).
\end{align}
Repeating the argument of Theorem~\ref{thm: main} with this correlation function yields the improved bound $\epsilon \le O(1/L^2)$.

\section{RK Hamiltonians with conserved quantities}

This section focuses on gapless RK Hamiltonians with conserved quantities.
In such systems, the classical configuration space decomposes into disconnected sectors that cannot be connected by the time evolution.
Each of these sectors possesses a degenerate ground state, and the Hamiltonian can be block-diagonalized.
As a result, the inequality for the dynamical exponent, $z \ge 2$, applies separately to each block.
We illustrate these ideas with two representative cases: (i) systems possessing a conserved $\mathrm{U}(1)$ charge and (ii) systems characterized by topological charges.

\subsection{FF Hamiltonians with invariant subspaces}

We first introduce a general framework for FF Hamiltonians that possess invariant subspaces.
The discussion applies to arbitrary FF systems and is not restricted to the RK Hamiltonians.
If local terms of a Hamiltonian have a common invariant subspace, we can consider the FF property and the Gosset-Huang inequality for each subspace.

Let $\mathcal{K}$ be a subspace of the Hilbert space, and let $\hat\Pi^{\mathcal{K}}$ be the orthogonal projector onto $\mathcal{K}$.
If the local terms $\hat{H}_i$ of the Hamiltonian satisfy
\begin{align}
 \forall i,\quad\hat H_i\hat\Pi^{\mathcal{K}} = \hat\Pi^{\mathcal{K}} \hat H_i \hat\Pi^{\mathcal{K}}, 
\end{align}
then $\mathcal{K}$ is an invariant subspace.
We define the Hamiltonian $\hat{H}^{\mathcal{K}}$ for an invariant subspace $\mathcal{K}$ as
\begin{align}
    \hat H_i^{\mathcal{K}} \coloneqq \hat\Pi^{\mathcal{K}} \hat H_i \hat\Pi^{\mathcal{K}},\quad \hat H^{\mathcal{K}} \coloneqq \sum_i \hat H_i^{\mathcal{K}}.
\end{align}
Although the supports of local terms $\hat H_i^{\mathcal{K}}$ are generally nonlocal, $\hat H^{\mathcal{K}}$ are still local in the sense that
\begin{align}
 [\hat H_i^\mathcal{K},\hat H_j^\mathcal{K}] = \hat\Pi^{\mathcal{K}}[\hat H_i,\hat H_j]\hat\Pi^{\mathcal{K}} = 0,
\end{align}
for sufficiently distant $i$ and $j$.
If there exists a ground state $|\Psi^{\mathcal{K}}\rangle \in \mathcal{K}$ of $\hat H^{\mathcal{K}}$, we say that $\hat H$ is FF for the invariant subspace $\mathcal{K}$.
In this case, the Gosset-Huang inequality for the invariant subspace $\mathcal{K}$ is written as
\begin{align}&
 \frac{\langle\mathcal{O}|(\hat\Pi^{\mathcal{K}}-\hat G^{\mathcal{K}})|\mathcal{O}'\rangle}{\||\mathcal{O}\rangle\| \||\mathcal{O}'\rangle\|} \nonumber\\
    &
 \le 2 \exp\left(-\left( \frac{D(|\mathcal{O}\rangle,|\mathcal{O}'\rangle)-1}{c-1}-2\right) \sqrt{\frac{\epsilon_{\mathcal{K}}}{g^2+\epsilon_{\mathcal{K}}}}\right),
    \label{GH inequality for invariant subspace}
\end{align}
where $\hat G^{\mathcal{K}}$ is the projector onto the ground space of $\hat H^{\mathcal{K}}$, and $\epsilon_{\mathcal{K}}$ represents the dimensionless spectral gap of $\hat H^{\mathcal{K}}$.

We can define the dynamical exponent $z_\mathcal{K}$ for an invariant subspace $\mathcal{K}$ by $\epsilon_\mathcal{K} \sim L^{-z_\mathcal{K}}$.
Theorem~\ref{thm: main} implies that $z_\mathcal{K} \ge 2$ if we can find excitations $|\mathcal{O}\rangle$ and $|\mathcal{O}'\rangle$ such that the left-hand side of Eq.~\eqref{GH inequality for invariant subspace} shows power-law decay.

\label{sec: non-ergodic RK}
\subsection{Degeneracy of the ground states} 
Let us consider the degeneracy of the ground states of an RK Hamiltonian.
Let $\mathcal{G}$ be a graph where vertices are labeled by the classical configurations $C$ and there is an edge between $C$ and $C'$ if $\langle C'|\hat H|C\rangle < 0$.
If $\mathcal{G}$ is connected, the Perron-Frobenius theorem shows that the Hamiltonian has a unique ground state.

If $\mathcal{G}$ is not connected, the Perron-Frobenius theorem applies to each connected component $\mathcal{G}^{(n)}$.
The Hilbert space can be decomposed into a direct sum as
\begin{align}
 \mathcal{H} = \bigoplus_n \mathcal{H}^{(n)},\quad \mathcal{H}^{(n)} = \Span\{|C^{(n)}\rangle\}_{C^{(n)} \in \mathcal{G}^{(n)}}.
\end{align}
The Hamiltonian $\hat{H}$ and the time evolution operator ${{e}}^{-\hat{H} t}$ are block diagonalized by this decomposition as, respectively,
\begin{align}
 \hat{H} = \bigoplus_n \hat H^{(n)},\quad {{e}}^{-\hat{H}t} = \bigoplus_n {{e}}^{-\hat{H}^{(n)}t}.
\end{align} 
Since $\hat{H}$ is positive semidefinite, $\hat{H}^{(n)}$ are also positive semidefinite.
We can easily construct the ground state of $\hat{H}^{(n)}$ with zero eigenvalue as
\begin{align}
 |\Psi_\mathrm{RK}^{(n)}\rangle \coloneqq \frac1{\sqrt{Z^{(n)}}} \sum_{C^{(n)}} \sqrt{w(C^{(n)})}|C^{(n)}\rangle, \label{def of n-th ground state of RK Hamiltonian}
\end{align}
where
\begin{align}
 Z^{(n)} \coloneqq \sum_{C^{(n)}} w(C^{(n)}).
\end{align}
Using the Perron-Frobenius theorem for each ${{e}}^{-\hat H^{(n)}t}$, we find that there are no other ground states than $|\Psi_\mathrm{RK}^{(n)}\rangle$.

If an arbitrary configuration is obtained with multiple local updates from a particular configuration, we can construct an ergodic local RK Hamiltonian with a unique ground state.
However, if there are conserved charges or topological invariants that remain unchanged under any local updates, the configuration space and Hilbert space decompose into distinct charge sectors or topological sectors, each with its ground state.
For instance, the spin-1/2 ferromagnetic Heisenberg model can be interpreted as an RK Hamiltonian describing the diffusion process of a 1D lattice gas that conserves the total particle number, known as the symmetric simple exclusion process.
This model has the ground-state degeneracy for each sector defined by the total particle number (or, equivalently, the total magnetization).

Another example is the Rokhsar-Kivelson point of the quantum dimer model~\cite{rokhsarSuperconductivityQuantumHardCore1988}, which has degenerate ground states classified by a topological invariant that takes values in the $\mathbb{Z}_2$ homology group.
In such cases, we restrict the configuration space or the Hilbert space to a sector with a fixed charge or topological invariant to use Theorem~\ref{thm: main}.
Within each sector, the RK Hamiltonian is ergodic, and the ground state becomes unique.

\subsection{RK Hamiltonians with conserved charges}
Let us consider the translationally invariant RK Hamiltonians with a $\mathrm{U}(1)$ symmetry.
This class includes the ferromagnetic Heisenberg model and the Kawasaki dynamics~\cite{kawasakiDiffusionConstantsCritical1966} in the disordered phase of the Ising model.
Let $\hat N = \sum_{i \in \Lambda} \hat n_i$ be the generator of global $\mathrm{U}(1)$ rotation, representing the particle number.
We assume that the classical configuration basis $\{|C\rangle\}$ are eigenstates of $\hat N$.
An RK Hamiltonian that conserves the particle number is obtained by eliminating the off-diagonal elements that change the particle number from an unconstrained RK Hamiltonian.
The ground states are given by
\begin{align}
 |\Psi_N\rangle = \sum_{C \in \mathcal{C}_N}\sqrt{\frac{w(C)}{Z_N}}|C\rangle,\quad
 Z_N = \sum_{C \in \mathcal{C}_N} w(C),
\end{align}
where $\mathcal{C}_N$ is the set of configurations with $N$ particles and $Z_N$ is the partition function for the $N$-particle sector.
We consider the infinite-volume limit $|\Lambda| \to \infty$ while fixing the density $\rho \coloneqq N/|\Lambda|$.

For computational convenience, we consider the ground state corresponding to the grand canonical distribution
\begin{align}
 |\Psi(\mu)\rangle = \sum_N \sqrt{\frac{Z_N{{e}}^{\beta\mu N}}{\Xi}} |\Psi_N\rangle,\quad
 \Xi = \sum_N Z_N {{e}}^{\beta\mu N},
\end{align}
where $\beta$ is the inverse temperature, $\mu$ is the chemical potential, and $\Xi$ is the grand partition function.
Let $\langle \cdot \rangle \coloneqq \langle\Psi(\mu)|\cdot|\Psi(\mu)\rangle$ and $\langle \cdot \rangle_N \coloneqq \langle\Psi_N|\cdot|\Psi_N\rangle$.
We assume that the density-density correlation function for the grand canonical distribution decays exponentially:
\begin{align}
 \langle n_i n_j \rangle - \langle n_i \rangle\langle n_j \rangle \approx C{{e}}^{-D(\hat n_i,\hat n_j)/\xi}.
    \label{assumption of cluster property}
\end{align}
Otherwise, the model becomes gapless without imposing the conservation law.
The relation between the density-density correlation function for the grand canonical distribution and the canonical distribution is given by
\begin{align}&
\langle n_i n_j \rangle = \langle \langle n_in_j \rangle_N \rangle \nonumber\\
&
= \sum_{m=0}^{\infty} \frac{\langle(N-\langle N \rangle)^m\rangle}{m!}
\left.\frac{\partial^m}{\partial N^m}\langle n_in_j \rangle_{N}\right|_{N=\langle N\rangle} \nonumber\\
&
\approx \langle n_i n_j \rangle_{\langle N \rangle} + \frac{\langle N^2 \rangle-\langle N \rangle^2}{2}
\frac{\partial^2}{\partial \langle N\rangle^2}\langle n_in_j \rangle_{\langle N\rangle} \nonumber\\
&
\approx \langle n_i n_j \rangle_{\langle N \rangle} + \frac{\langle N^2 \rangle-\langle N \rangle^2}{2}
\frac{\partial^2}{\partial \langle N\rangle^2}\langle n_in_j \rangle,
\label{relation between canonical and grand canonical}
\end{align}
where we dropped the higher-order terms.
We also assume that the translation symmetry is not spontaneously broken.
Then, we have $\langle n_i \rangle = \langle N \rangle/|\Lambda| = \rho$.
The dominant contribution to $\langle n_i n_j \rangle_{\langle N\rangle}$ is given by $\langle n_i \rangle\langle n_j \rangle = \langle N \rangle^2/|\Lambda|^2$ because of the cluster property in Eq.~\eqref{assumption of cluster property}.
Thus, the second term of Eq.~\eqref{relation between canonical and grand canonical} is evaluated as
\begin{align}
 \frac{\langle N^2 \rangle-\langle N \rangle^2}{2}\frac{\partial^2}{\partial\langle N \rangle^2} \langle n_i n_j \rangle
    &
 \approx \frac{\langle N^2 \rangle-\langle N \rangle^2}{|\Lambda|^2} \nonumber\\
    &
 = \frac{1}{\beta|\Lambda|}\frac{\partial\rho(\mu,\beta)}{\partial\mu},
\end{align}
where $\partial\rho/\partial\mu$ is the compressibility.
Therefore, we have
\begin{align}&
 |\langle\Psi_{\langle N \rangle}|\hat n_i(\hat{\mathbbm{1}}-\hat G)\hat n_j|\Psi_{\langle N \rangle}\rangle| \nonumber\\
    &
 = \left|\langle n_i n_j \rangle_{\langle N \rangle} - \langle n_i \rangle_{\langle N \rangle}\langle n_j \rangle_{\langle N \rangle} \right| \nonumber\\
    &
 \approx \left| \langle n_i n_j \rangle - \langle n_i \rangle\langle n_j \rangle - \frac{1}{\beta|\Lambda|}\frac{\partial\rho(\mu,\beta)}{\partial\mu}\right| \nonumber\\
    &
 = O\left( \frac1{|\Lambda|} \right).
\end{align}
Here, we used Eqs.~\eqref{assumption of cluster property} and~\eqref{relation between canonical and grand canonical}.
Since the normalization constant is $O(1)$, the normalized correlation function also decays as $O(1/|\Lambda|)$.
Therefore, the dynamical exponent satisfies $z \ge 2$ for all particle-number sectors except for the points where the compressibility vanishes.


\subsection{The RK point of the quantum dimer model}
\label{sec: RK point}

As an example of the RK Hamiltonian with topological degeneracy, we focus on the RK Hamiltonian of the classical dimer model, also known as the RK point of the quantum dimer model~\cite{rokhsarSuperconductivityQuantumHardCore1988}.
While this model is the origin of the RK Hamiltonian, it is a comparatively intricate example in that its Hilbert space is not a tensor product of local Hilbert spaces, and its ground states are degenerate.

Let us consider the 2D square lattice with linear size $L$ with periodic boundary conditions.
Let $l_e$ be a variable associated with the edge $e$, which takes the values of zero or one.
A classical configuration $C = \{l_e\}$ is called a dimer configuration if
\begin{align}
 l_{(x+1/2, y)}+l_{(x-1/2, y)}+l_{(x, y+1/2)}+l_{(x, y-1/2)} = 1 \label{dimer constraint}
\end{align}
for all vertices $(x, y) \in \mathbb{Z}_L^2$.
Here, the edges are labeled by the position of the midpoint.
We illustrate examples of dimer configurations in Fig.~\ref{fig: dimerConfiguration}, where solid edges represent $l=1$ and empty edges represent $l=0$.
The dimer constraint implies that precisely one solid line ends at each vertex.
The classical dimer model is defined by the equal weight probability distribution for all dimer configurations.

\begin{figure}[b]
    \centering
    \includegraphics[width=0.35\hsize]{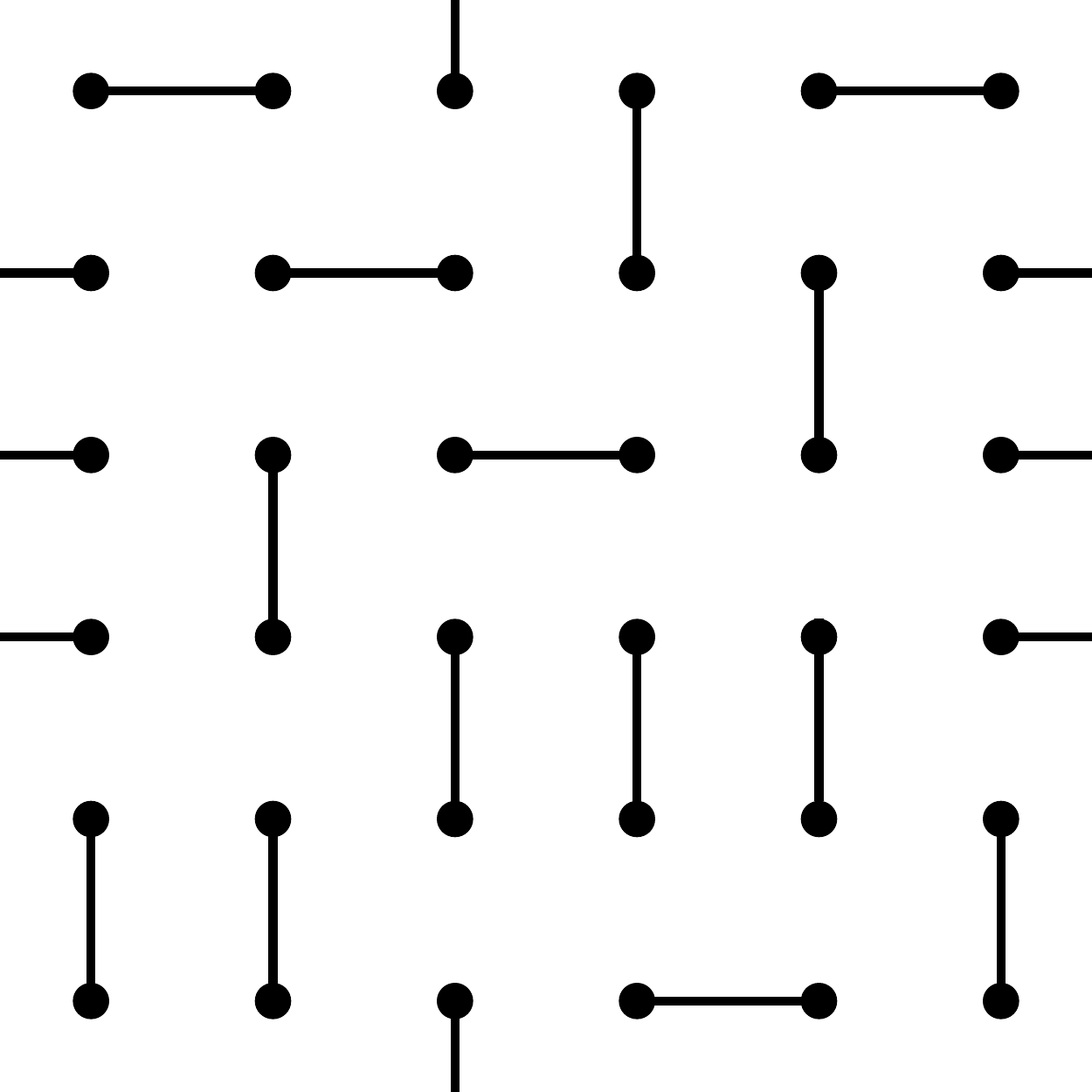}
    \caption{A dimer configuration on the square lattice.
    \label{fig: dimerConfiguration}}
\end{figure}

We consider the Hilbert space
\begin{align}
 \mathcal{H} = \bigotimes_{e\text{: edges}} \mathcal{H}_e,\quad
 \mathcal{H}_e = \Span\{|0_e\rangle, |1_e\rangle\}.
\end{align}
The Hamiltonian of the RK point of the quantum dimer model is given by
\begin{align}
 \hat H = \sum_{p\text{: plaquettes}}\hat H_p + \sum_{v\text{: vertices}} \alpha\hat\Pi_v.
 \label{Hamiltonian of RK quantum dimer}
\end{align}
Here, $\hat H_p$ is defined for each plaquette as
\begin{align}
 \hat H_p &\coloneqq \frac1{2}\left(|\Dimerpara\rangle-|\Dimerperp\rangle\right)
 \left(\langle\Dimerpara|-\langle\Dimerperp|\right).
\end{align}
The term $\hat{\Pi}_v$ is the projector onto the subspace that violates the dimer constraint in Eq.~\eqref{dimer constraint} around the vertex $v$, and $\alpha$ is a positive parameter.
By taking the limit $\alpha \to +\infty$, we restrict the Hilbert space to the subspace given by
\begin{align}
 \mathcal{H}_D = \Span\{|\{l_e\}\rangle \mid \{l_e\} \in D \},
\end{align}
where $D$ is the set of all dimer configurations.
The Hamiltonian for the restricted Hilbert space is given by
\begin{align}
    \hat H_D \coloneqq \sum_{p\text{: plaquettes}} \hat\Pi_D\hat H_p \hat\Pi_D,
    \label{Hamiltonian of modified RK quantum dimer}
\end{align}
where $\hat\Pi_D$ is the projector onto $\mathcal{H}_D$.
While the tensor product structure of the Hilbert space is broken by the dimer constraint, the Hamiltonian is still local in the sense that distant local terms commute with each other:
\begin{align}
 [\hat\Pi_D\hat H_p\hat\Pi_D, \hat\Pi_D\hat H_{p'}\hat\Pi_D]
    &
 = \hat\Pi_D[\hat H_p, \hat H_{p'}]\hat\Pi_D  \nonumber\\
    &
 = 0 \quad (p \cap p' = \emptyset).
    \label{locality of RK point}
\end{align}
Here, we used $\hat\Pi_D\hat H_p\hat\Pi_D = \hat\Pi_D\hat H_p = \hat H_p\hat\Pi_D$.
The Hamiltonian in Eq.~\eqref{Hamiltonian of modified RK quantum dimer} is an RK Hamiltonian since one of the ground states is written as
\begin{align}
 |\Psi_D\rangle \coloneqq \frac1{\sqrt{|D|}} \sum_{C \in D} |C\rangle,
 \label{GS of RK dimer}
\end{align}
where $|D|$ is the number of dimer configurations.
The Boltzmann weight $w(C)$ is given by that of the classical dimer model: $w(C) = 1$.
Therefore, the transition-rate matrix of the corresponding Markov process is given by $\hat{W} = -\hat{H}$.

The configuration space is decomposed into topological sectors as well as a set of isolated staggered configurations.~\cite{roisingErgodicArchimedeanDimers2023} illustrated in Fig.~\ref{fig: dimerConfigurations}.
Any state in the same topological sector can be obtained by repeating local transitions $\Dimerpara \leftrightarrow \Dimerperp$.
Staggered configurations depicted in Fig.~\ref{fig: staggered configuration} are isolated and do not admit any local transition.
\begin{figure}[t]
    \centering
    \begin{minipage}{0.3\hsize}
        \centering
        \includegraphics[width=0.8\hsize]{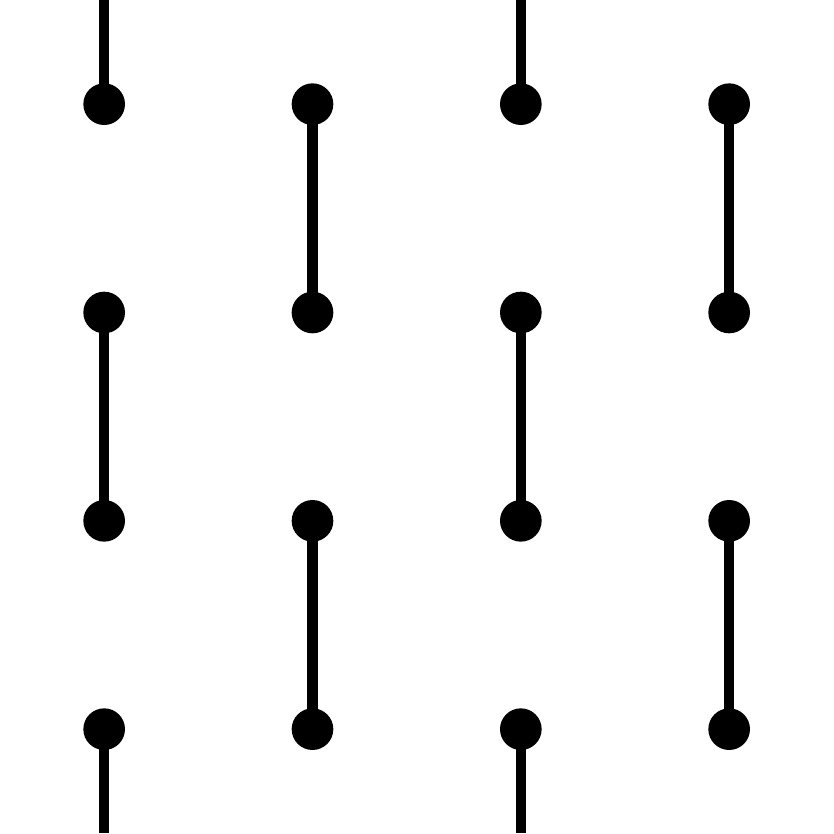}
        \subcaption{\label{fig: staggered configuration}}
    \end{minipage}
    \begin{minipage}{0.3\hsize}
        \centering
        \includegraphics[width=0.8\hsize]{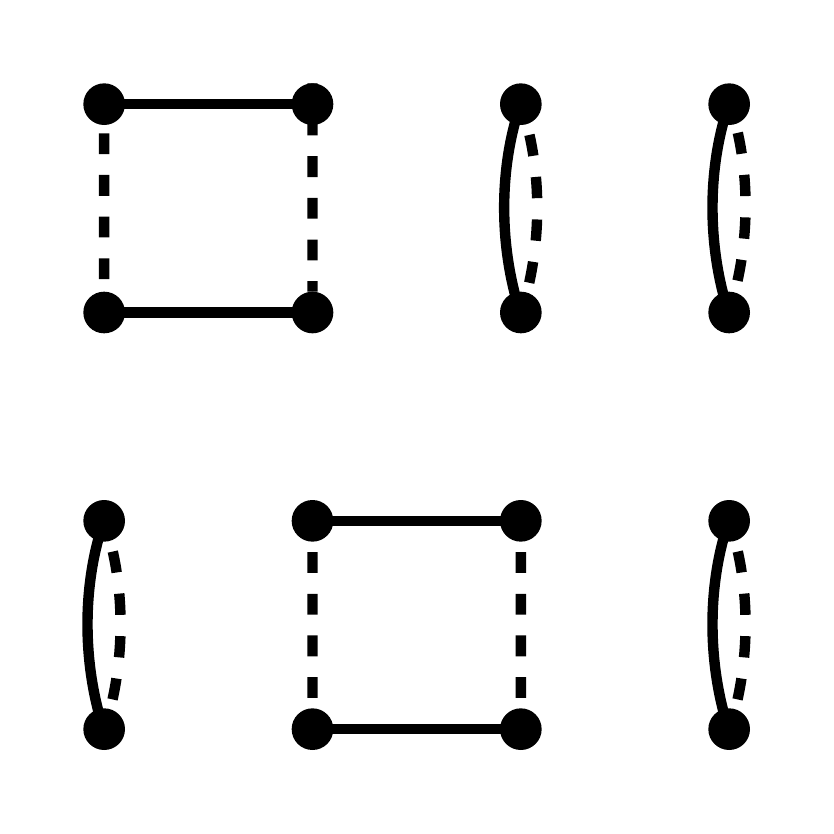}
        \subcaption{\label{fig: trivial configuration}}
    \end{minipage}
    \begin{minipage}{0.3\hsize}
        \includegraphics[width=0.8\hsize]{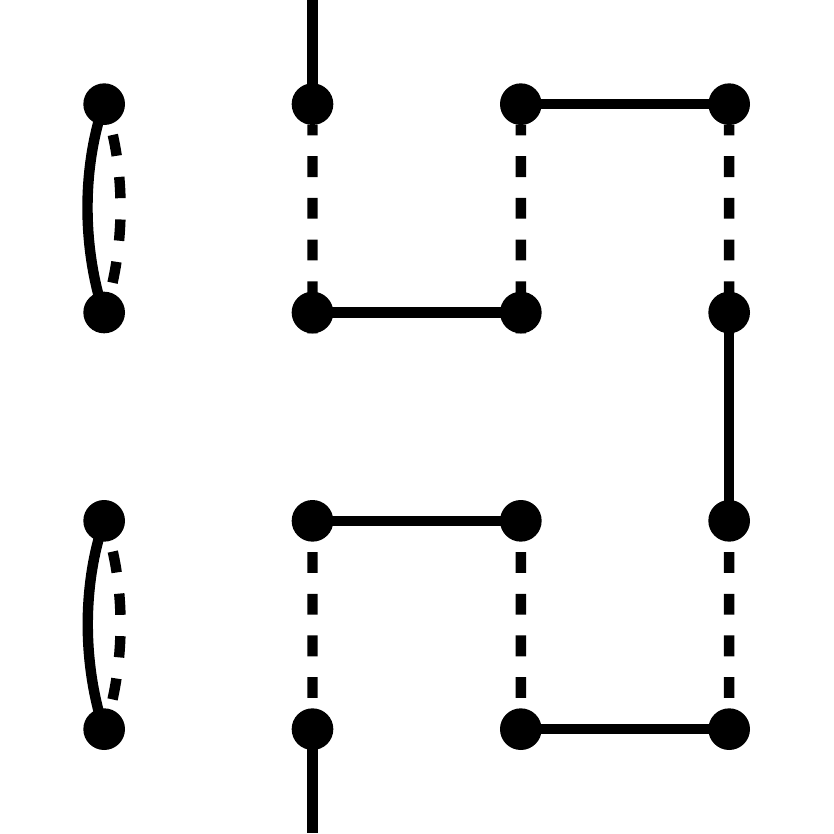}
        \subcaption{\label{fig: nontrivial configuration}}
    \end{minipage}
\caption{(a) A staggered configuration.
(b) A dimer configuration (solid line) in the same topological sector as the reference configuration (dashed line).
(c) A dimer configuration (solid line) in a topological sector that is different from the reference configuration (dashed line).\label{fig: dimerConfigurations}}
\end{figure}
Each topological sector $\Omega$ has a unique ground state given by
\begin{align}
 |\Psi_\Omega\rangle \coloneqq \frac1{\sqrt{|\Omega|}} \sum_{C \in \Omega} |C\rangle.
\end{align}
The states corresponding to staggered configurations are also ground states by themselves.
We can define an ergodic Hamiltonian for each topological sector $\Omega$ by projecting the Hamiltonian onto the subspace $\Span\{|C\rangle \mid C \in \Omega\}$.
This Hamiltonian is also local, as shown in Eq.~\eqref{locality of RK point}.
The Hamiltonians for each staggered configuration are trivial because the Hilbert space is one-dimensional.

Excluding staggered configurations, correlation functions on the ground states exhibit power-law decay.
For example, the asymptotic behavior of the correlation function for two horizontal edges in the same row is given by~\cite{fisherStatisticalMechanicsDimers1963}
\begin{align}
 |\langle\Psi_\Omega|\hat l_{(1/2, 0)}(\hat{\mathbbm{1}}-|\Psi_\Omega\rangle\langle\Psi_\Omega|) \hat l_{(r+1/2, 0)}|\Psi_\Omega\rangle |
 \approx  \frac{\mathrm{const}}{r^2},
\end{align}
where $\hat l_{(x, y)}$ is the operator that gives the filling number (0 or 1) of the edge whose midpoint is at $(x, y) \in \mathbb{R}^2$.
Therefore, our bound $z \ge 2$ holds for each topological sector.

A numerical study shows $z = 2.01(2)$ for the RK point in the square lattice~\cite{isakovDynamicsConformalQuantum2011}.
We can also obtain the exact dynamical exponent $z = 2$ from the effective field theory.
Dimer configurations on the square lattice can be described by height variables~\cite{henleyRelaxationTimeDimer1997}.
In the continuum limit, the height variable description is identified with the effective field theory called the quantum Lifshitz model~\cite{henleyRelaxationTimeDimer1997, ardonneTopologicalOrderConformal2004}, whose action is given by
\begin{align}
 S = \int\mathinner{\mathrm{d}^3x} \left(\frac1{2}(\partial_t\varphi)^2 + \frac{\kappa^2}{2}(\nabla^2\varphi)^2\right),
\end{align}
where $\varphi$ is the compactified scalar field.
The RK point corresponds to $\kappa = 1/(2\pi)$.
The canonical dimensional analysis yields the exact value of the dynamical exponent $z = 2$.

\section{Lower bound by plane-wave variational state}\label{sec: variational bound}
This section gives a lower bound on the dynamical exponent for translation-invariant gapless FF systems by plane-wave variational state.
This approach yields the bound $z\ge 2$ for systems that possess a plane-wave ground state, a typical example being systems with spontaneous breaking of a continuous symmetry.
For the RK Hamiltonians of critical points, the same variational method yields the bound $z \ge 2-\eta$, where $\eta$ is the anomalous dimension of the corresponding classical critical point.
This bound coincides with previously established rigorous results in Refs.~\cite{abeDynamicsIsingModel1969,halperinRigorousInequalitiesSpinRelaxation1973,lubetzkyCriticalIsingSquare2012}.

Let us consider a translationally invariant FF Hamiltonian $\hat{H} = \sum_i \hat{H}_i$ on a $d$-dimensional lattice with a ground state $|\Psi\rangle$.
We consider the following plane-wave state.
\begin{align}
 |\mathcal{O}(\bm{k})\rangle = \sum_i {{e}}^{{{i}}\bm{k}\cdot\bm{x}_i}\hat{\mathcal{O}}_i |\Psi\rangle.
\end{align}
We assume $\langle\Psi|\hat{\mathcal{O}}_i|\Psi\rangle = 0$ by subtracting a constant from $\hat{\mathcal{O}}_i$.
Then, the variational state $|\mathcal{O}(\bm{k})\rangle$ is orthogonal to the ground state $|\Psi\rangle$.
We assume that $|\mathcal{O}(\bm{k})\rangle$ is also orthogonal to other ground states except for a finite number of momentum values $\bm{k} = \bm{k}_1, \bm{k}_2, \ldots, \bm{k}_n$.
Each local term $\hat{H}_i$ of the Hamiltonian can be transformed into a projector without altering the ground-state subspace or the dynamical exponent of the system.
Assuming that each $\hat{H}_i$ is a projector, the energy expectation of this state is evaluated as
\begin{align}
 \frac{\langle\mathcal{O}(\bm{k})|\hat H|\mathcal{O}(\bm{k})\rangle}{\langle\mathcal{O}(\bm{k})|\mathcal{O}(\bm{k})\rangle}
 &
 = \sum_i \frac{\langle\mathcal{O}(\bm{k})|\hat H_i^2|\mathcal{O}(\bm{k})\rangle}{\langle\mathcal{O}(\bm{k})|\mathcal{O}(\bm{k})\rangle} \nonumber\\
 &
 = L^d\frac{\| \hat H_i|\mathcal{O}(\bm{k})\rangle\|^2}{\||\mathcal{O}(\bm{k})\rangle\|^2} \nonumber\\
 &
 = \frac{\| \sum_j {{e}}^{{{i}}\bm{k}\cdot\bm{x}_j}\hat H_i\hat{\mathcal{O}}_j|\Psi\rangle \|^2}{\||\mathcal{O}(\bm{k})\rangle\|^2/L^d}\nonumber\\
 &
 = \frac{\| \sum_{j} {{e}}^{{{i}}\bm{k}\cdot\bm{x}_j}[\hat H_i,\hat{\mathcal{O}}_j]|\Psi\rangle \|^2}{S_{\mathcal{O}}(\bm{k})}, \label{energy expectation of single-particle excitation}
\end{align}
where we used FF condition $\hat{H_i}|\Psi\rangle = 0$ and $S_{\mathcal{O}}(\bm{k})$ is a correlation function in momentum space defined as
\begin{align}
 S_{\mathcal{O}}(\bm{k}) \coloneqq \frac{\||\mathcal{O}(\bm{k})\rangle\|^2}{L^d} = \sum_j {{e}}^{{{i}}\bm{k}\cdot(\bm{x}_j-\bm{x}_i)} \langle\Psi|\mathcal{O}_i^\dagger \mathcal{O}_j|\Psi\rangle.
 \label{def of G}
\end{align}
The variational state $|\mathcal{O}(\bm{k})\rangle$ is a gapless excitation if the numerator of Eq.~\eqref{energy expectation of single-particle excitation} converges to zero at $|\bm{k}-\bm{k}_0| \to 0$, or the denominator diverges in the thermodynamic limit.

First, we consider the case that $|\mathcal{O}(\bm{k}_0)\rangle$ is a ground state.
For simplicity, we take $\bm{k}_0 = \bm{0}$.
We have
\begin{align}
 \hat H_i |\mathcal{O}(\bm{0})\rangle = \sum_{j} [\hat H_i,\hat{\mathcal{O}}_j]|\Psi\rangle = 0.
\label{energy of NG modes}
\end{align}
We can evaluate the numerator of Eq.~\eqref{energy expectation of single-particle excitation} using Eq.~\eqref{energy of NG modes} as
\begin{align}&
 \| \sum_{j} {{e}}^{{{i}}\bm{k}\cdot\bm{x}_j}[\hat H_i,\hat{\mathcal{O}}_j]|\Psi\rangle \|^2 \nonumber\\
 &
 = \| \sum_{j} ({{e}}^{{{i}}\bm{k}\cdot\bm{x}_j}-{{e}}^{{{i}}\bm{k}\cdot\bm{x}_i})[\hat H_i,\hat{\mathcal{O}}_j]|\Psi\rangle \|^2 \nonumber\\
 &
 = O(|\bm{k}|^2).
\end{align}
Here, we used the fact that $[\hat H_i,\hat{\mathcal{O}}_j]$ is nonzero only if $\bm{x}_j$ is in a finite region around $\bm{x}_i$.
The state $|\mathcal{O}(\bm{k})\rangle$ is orthogonal to ground states if $\bm{k} \ne \bm{0}$ and $\bm{k}$ are sufficiently small.
Thus, we obtain
\begin{align}
 \epsilon \le \frac{\langle\mathcal{O}(\bm{k})|\hat H|\mathcal{O}(\bm{k})\rangle}{\langle\mathcal{O}(\bm{k})|\mathcal{O}(\bm{k})\rangle} = \frac{O(|\bm{k}|^2)}{S_{\mathcal{O}}(\bm{k})}.
 \label{lower bound by Nambu Goldstone modes}
\end{align}
To obtain $\epsilon \le O(|\bm{k}|^2)$, we have to assume that
\begin{align}
    \lim_{\bm{k} \to \bm{0}} S_{\mathcal{O}}(\bm{k}) \ne 0.
\end{align}

Next, we consider the case of the RK Hamiltonians of critical points with a unique ground state.
Let $\hat{\mathcal{O}}_i$ be a real diagonal operator that corresponds to a primary field in the associated conformal field theory, excluding the identity operator.
Since the ground state $|\Psi\rangle$ is critical, $S_{\mathcal{O}}(\bm{0})$ diverges as
\begin{align}
 S_{\mathcal{O}}(\bm{0}) = \sum_j \langle\mathcal{O}_i\mathcal{O}_j\rangle \propto L^{d-2\Delta_{\mathcal{O}}},
\end{align}
where $\Delta_{\mathcal{O}}$ is the scaling dimension of $\hat{\mathcal{O}}_i$ and $L$ is the linear system size.
Thus
\begin{align}
 \frac{\langle\mathcal{O}(\bm{0})|\hat H|\mathcal{O}(\bm{0})\rangle}{\langle\mathcal{O}(\bm{0})|\mathcal{O}(\bm{0})\rangle} = \frac{O(L^d)}{L^d S_{\mathcal{O}}(\bm{0})} = O(L^{-(d-2\Delta_{\mathcal{O}})}).
\end{align}
Since $|\mathcal{O}(\bm{0})\rangle$ is orthogonal to ground states, we obtain
\begin{align}
 z \ge d - 2\min_{\mathcal{O} \ne 1}\Delta_{\mathcal{O}} = 2-\eta,
 \label{Halperin bound}
\end{align}
where $\eta$ is the anomalous dimension.
This bound is equivalent to the long-known rigorous bound $\Delta \ge \gamma$~\cite{abeDynamicsIsingModel1969,halperinRigorousInequalitiesSpinRelaxation1973}.
For example, the 2+1D Ising conformal quantum critical point satisfies $z \ge 7/4$, because the lowest scaling dimension of the primary operator in the 2D Ising conformal field theory, excluding the identity, is $1/8$.
The same bound $z \ge 7/4$ is obtained in Ref.~\cite{lubetzkyCriticalIsingSquare2012} using essentially the same approach.
A numerical calculation shows that $z = 2.1667(5)$~\cite{nightingaleMonteCarloComputation2000}.


For unitary conformal field theories and scalar operators, the unitarity bound is given~\cite{minwallaRestrictionsImposedSuperconformal1998} by
\begin{align} &
 \Delta_{\mathcal{O}} \ge \frac{1}{2}(d-2).
\label{unitarity bound}
\end{align}
Thus, we obtain $\eta \ge 0$.
This bound also holds for spinning operators.
Therefore, except when the equality in Eq.~\eqref{unitarity bound} holds, our result $z \ge 2$ improves the lower bound by plane wave variational states in the case of the RK Hamiltonians, whose ground state corresponds to a unitary conformal field theory.
We note that finite particle-number variational states also yield the same bound as Eq.~\eqref{Halperin bound}.
The fact that numerically observed values of $z$ are larger than $2-\eta$ suggests that the true low-energy excitations in these critical RK Hamiltonians are not welldescribed by these single-particle plane-wave states and may possess a more complex, nonlocal character.

\section{Parent Hamiltonians of critical projected entangled pair states}
\label{sec: PEPS}

Projected entangled pair states (PEPS)~\cite{verstraeteRenormalizationAlgorithmsQuantumMany2004,verstraeteCriticalityAreaLaw2006} represent a more general class of ground states of the FF Hamiltonian than those of RK Hamiltonians.
We note that the ground states of RK Hamiltonians in Eq.~\eqref{Boltzmann state} written in the Boltzmann weight can also be described as PEPS~\cite{verstraeteCriticalityAreaLaw2006}.
Importantly, PEPS does not require the wavefunction to be real. However, the effect of the complex phases of the ground states on the universal properties of quantum critical points is not yet fully understood.

To evaluate the correlation function in the Gosset-Huang inequality for PEPS, one must identify the complete set of ground states.
Therefore, establishing the bound $z \ge 2$ requires theoretical control over ground-state degeneracy.
Various conditions that determine the ground-state degeneracy of parent Hamiltonians are known~\cite{perez-garciaPEPSUniqueGround2007,schuchPEPSGroundStates2010,buerschaperTwistedInjectivityProjected2014,sahinogluCharacterizingTopologicalOrder2021}.
In this section, we review one such condition, known as injectivity~\cite{perez-garciaPEPSUniqueGround2007}.

\subsection{Definitions}
A PEPS is a quantum state specified by a graph and a tensor located at each vertex of the graph.
The tensor $A^{(i)}$ on a vertex $i$ has a physical index $s_i$ and virtual indices $\sigma_{\langle i,i_1\rangle},\ldots,\sigma_{\langle i,i_n\rangle}$, which correspond to the edges incident to the vertex $i$.
A PEPS $|\Psi(A)\rangle$ is given by
\begin{align}
 |\Psi(A)\rangle = \sum_{\{s\},\{\sigma\}} \prod_{i}
 A^{(i)}_{s_i \sigma_{\langle i,i_1\rangle} \cdots \sigma_{\langle i,i_n\rangle}} |\{s\}\rangle.
 \end{align}
Here, all pairs of virtual spins representing the same edge are contracted.
For example, PEPS for the square lattice is depicted as
\begin{align}
 |\Psi(A)\rangle = \vcenter{\hbox{\includegraphics[scale=0.13]{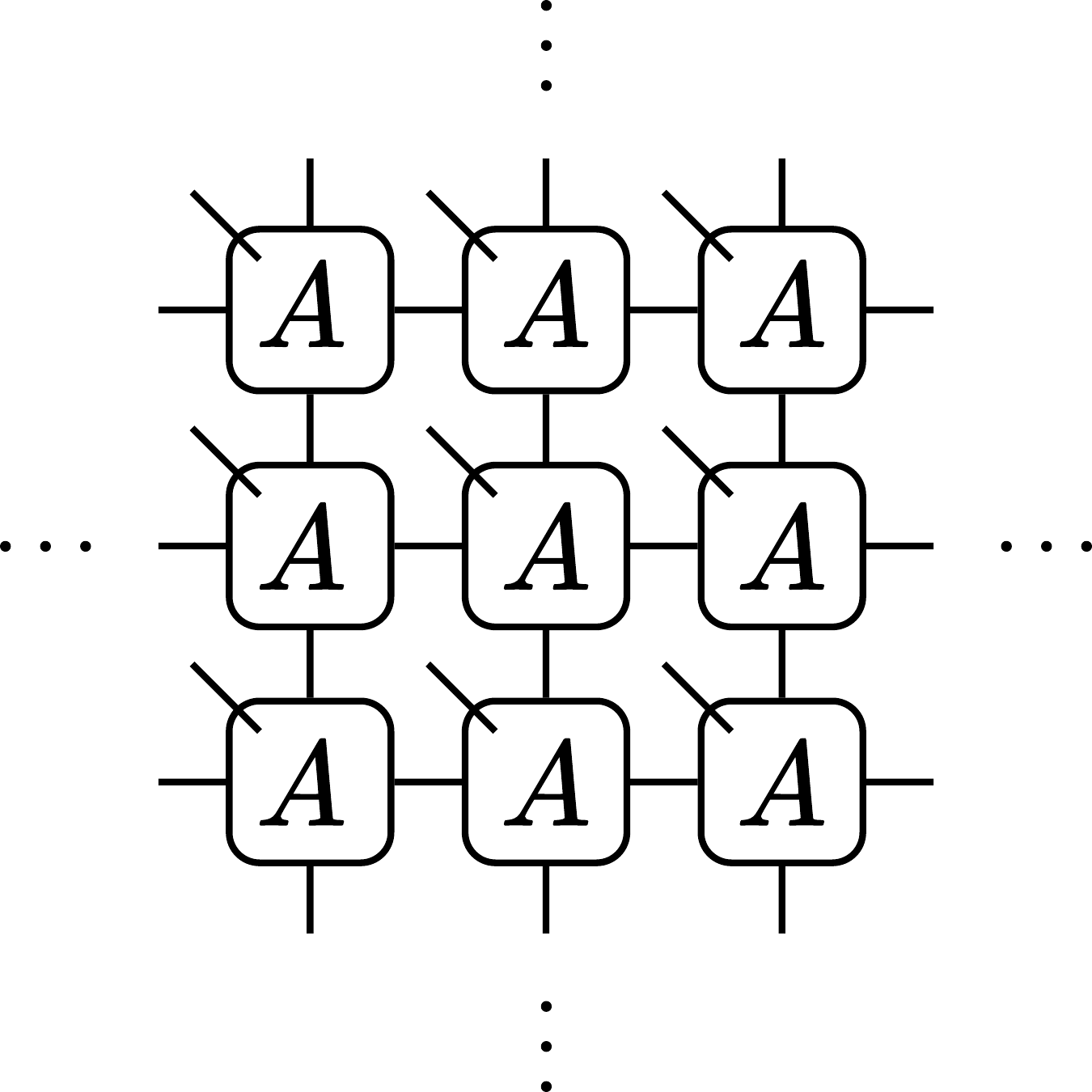}}},
\end{align}
where the diagonal lines represent physical spins and the horizontal and vertical lines represent virtual spins.
We note that each $A^{(i)}$ need not be the same tensor.

For a finite region $R_i$ around the point $i$, we consider the reduced density matrix defined as
\begin{align}
 \rho_{R_i} \coloneqq\operatorname{Tr}_{\mathcal{S}\setminus R_i}|\Psi(A)\rangle\langle\Psi(A)|,
\end{align}
where $\mathcal{S}\setminus R_i$ is the region excluding $R_i$ from the graph.
The local term $\hat H_i$ of the parent Hamiltonian for a PEPS $|\Psi(A)\rangle$ and an interaction range $R_i$ is the projector such that
\begin{align}
 \ker \hat H_i = \operatorname{supp}(\rho_{R_i}),
\end{align}
where $\operatorname{supp}(\rho_{R_i})$ is the support space of $\rho_{R_i}$.

\subsection{Injectivity}
Injectivity is a property that guarantees the uniqueness of the PEPS ground state.
For a region $R$ on the graph, the linear map $\Gamma_R$ is defined by
\begin{align}
 \Gamma_R(C) &\coloneqq \sum_{\{s_R\}}\sum_{\{\sigma_R\}}\sum_{\{\sigma_{\partial R}\}}C_{\{\sigma_{\partial R}\}} \nonumber\\
 &\quad
 \times\left(\prod_{i \in R}A^{(i)}\right)_{\{s_R\}\{\sigma_R\}\{\sigma_{\partial R}\}}|\{s_R\}\rangle,
\end{align}
where $\{s_R\}$ is the set of physical spins on $R$, $\{\sigma_{\partial R}\}$ is the set of virtual spins on the boundary of $R$, and $\{\sigma_R\}$ is the set of virtual spins inside $R$.
For example, if tensors are on a square lattice and $R$ is a plaquette consisting of four vertices, then $\Gamma_R(C)$ can be depicted as
\begin{align}
 \Gamma_{\vcenter{\hbox{\includegraphics[scale=0.07]{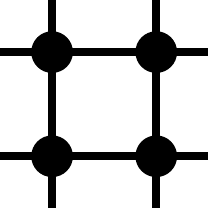}}}}(C) = \vcenter{\hbox{\includegraphics[scale=0.13]{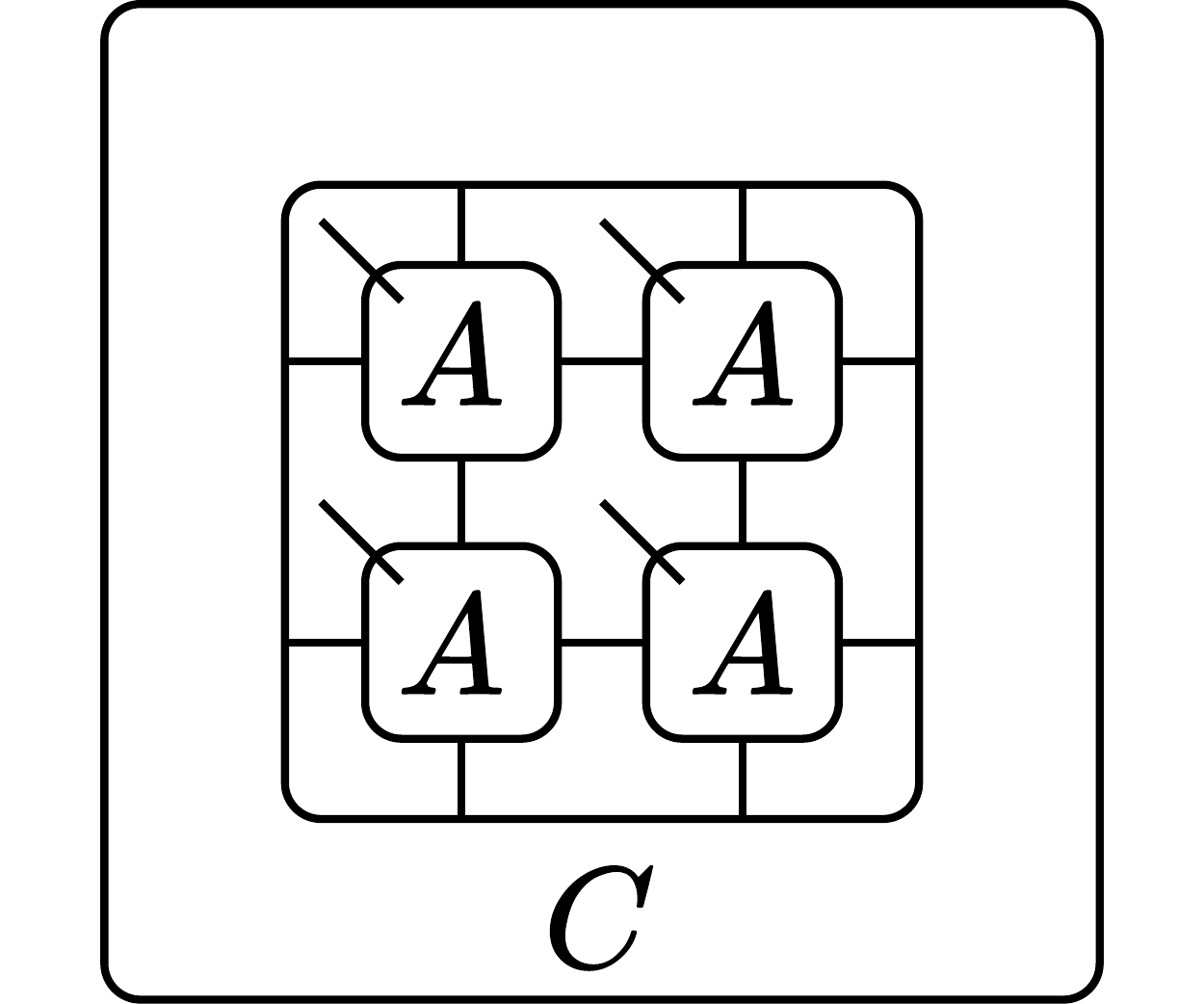}}}.
\end{align}
A region $R$ is called injective if $\Gamma_R$ is an injective map.
Also, a PEPS is injective if there is a proper covering of disjoint injective regions.
We consider a super-lattice such that each vertex is an injective region.
The ground state of the parent Hamiltonian with an injective PEPS is unique if the interaction range contains two adjacent vertices of the super-lattice~\cite{perez-garciaPEPSUniqueGround2007}.
This leads to the following corollary of Theorem~\ref{thm: main}.

\begin{cor}\label{cor: PEPS}
 For parent Hamiltonians of injective PEPS with a sufficiently large interaction range, if the ground state exhibits algebraic correlation functions, the dynamical exponent must satisfy $z \ge 2$.
\end{cor}

The rigorous bound $z \ge 2$ remains valid even in the presence of topological degeneracy.
For instance, the ground state of the RK point of the quantum dimer model in Eq.~\eqref{GS of RK dimer} can also be represented as PEPS.
This model exhibits both topological order and a critical ground state.

\bibliography{ref}

\end{document}